\documentclass[sigconf]{acmart}

\usepackage[T1]{fontenc}

\usepackage[utf8]{luainputenc}
\usepackage[english]{babel}

\usepackage{amsmath}
\usepackage{amssymb}
\usepackage{amsthm}
\usepackage{fixme}
\usepackage{xspace}
\usepackage{paralist}
\usepackage{bm}

\newcommand{\AlgReturn}[1]{ \textbf{return} #1}
\newcommand{\chan}{\mathcal{C}}

\DeclareMathOperator{\fout}{\mathsf{out}}
\DeclareMathOperator{\Fun}{\mathsf{Fun}}

\DeclareMathOperator{\Dist}{\mathsf{Dist}}
\DeclareMathOperator{\Eval}{\mathsf{Eval}}
\DeclareMathOperator{\fin}{\mathsf{in}}
\DeclareMathOperator{\RejSam}{\mathsf{RejSam}}
\DeclareMathOperator{\Gen}{\mathsf{Gen}}
\DeclareMathOperator{\AGen}{\mathsf{AGen}}

\DeclareMathOperator{\SGen}{\mathsf{SGen}}
\DeclareMathOperator{\SEnc}{\mathsf{SEnc}}
\DeclareMathOperator{\Steg}{\mathsf{S}}
\DeclareMathOperator{\negl}{\mathsf{negl}}
\DeclareMathOperator{\Ward}{\mathsf{Ward}}
\DeclareMathOperator{\Watch}{\mathsf{Watch}}

\DeclareMathOperator{\encwatch}{enc-watch}

\DeclareMathOperator{\ml}{\mathsf{ml}}
\DeclareMathOperator{\query}{\mathsf{query}}
\DeclareMathOperator{\sigl}{\mathsf{sl}}
\DeclareMathOperator{\outl}{\mathsf{ol}}
\DeclareMathOperator{\cl}{\mathsf{cl}}

\DeclareMathOperator{\CH}{CH}
\DeclareMathOperator{\ASAE}{\mathsf{ASA}}
\DeclareMathOperator{\ASA}{\mathsf{ASA}}

\DeclareMathOperator{\SES}{\mathsf{SES}}
\DeclareMathOperator{\SIG}{\mathsf{SIG}}
\newcommand{\combine}[2]{#1\!.\!#2}
\DeclareMathOperator{\RAND}{RAND}
\DeclareMathOperator{\SDec}{\mathsf{SDec}}
\newcommand{\nats}{\mathbb{N}}
\DeclareMathOperator{\unrel}{\mathbf{UnRel}}
\DeclareMathOperator{\insec}{\mathbf{InSec}}
\DeclareMathOperator{\InSec}{\mathbf{InSec}}
\DeclareMathOperator{\adv}{\mathbf{Adv}}
\DeclareMathOperator{\Adv}{\mathbf{Adv}}
\DeclareMathOperator{\cha}{cha}
\DeclareMathOperator{\cpa}{cpa}
\DeclareMathOperator{\sig}{sig}
\DeclareMathOperator{\supp}{supp}
\DeclareMathOperator{\Sign}{\mathsf{Sign}}
\DeclareMathOperator{\Att}{\mathsf{A}}
\DeclareMathOperator{\AExt}{\mathsf{AExt}}
\DeclareMathOperator{\Ext}{\mathsf{Ext}}
\DeclareMathOperator{\algf}{\mathsf{F}}

\DeclareMathOperator{\Vrfy}{\mathsf{Vrfy}}
\DeclareMathOperator{\Forg}{\mathsf{Fo}}
\DeclareMathOperator{\Enc}{\mathsf{Enc}}
\DeclareMathOperator{\AEnc}{\mathsf{AEnc}}
\DeclareMathOperator{\Dec}{\mathsf{Dec}}
\DeclareMathOperator{\Sig-Forge}{\mathsf{Sig-Forge}}
\DeclareMathOperator{\CPA-Dist}{\mathsf{CPA-Dist}}
\DeclareMathOperator{\SS-CHA-Dist}{\mathsf{SS-CHA-Dist}}
\DeclareMathOperator{\EncASA-Dist}{\mathsf{ASA-Dist}}
\DeclareMathOperator{\RASA-Dist}{\mathsf{RASA-Dist}}

\newcommand{\minent}{H_{\infty}}

\DeclareMathOperator{\R}{\mathsf{R}}
\DeclareMathOperator{\AR}{\mathsf{AR}}
\DeclareMathOperator{\prf}{prf}

\DeclareMathOperator{\Geni}{\mathsf{GenInput}}

\usepackage[plain]{algorithm}

\usepackage[noend]{algpseudocode}
\algrenewcommand\algorithmicrequire{\textbf{Parties:}} 
\algrenewcommand\algorithmicensure{\textbf{Input:}} 
\algrenewcommand\algorithmicindent{2em}%

\usepackage{tcolorbox}

\newcommand{\myAlgorithm}[5]{
  \begin{algorithm}[H]
  \caption{#1: #5}
  \begin{tcolorbox}[title={#1}, colframe=black!10, coltitle=black]
    \begin{algorithmic}[1]
      \Require #3
      \Ensure #2
      #4
    \end{algorithmic}
  \end{tcolorbox}
  \end{algorithm}
}

\newcommand{\inputAlgorithm}[4]{
  \begin{algorithm}[H]
  \caption{#4}
  \begin{tcolorbox}[title={#1}, colframe=black!10, coltitle=black]
    \begin{algorithmic}[1]
      \Ensure #2
      #3
    \end{algorithmic}
  \end{tcolorbox}
  \end{algorithm}
}

\newcommand{\mlAlgorithm}[4]{
  \begin{algorithm}[H]
  \caption{#4}
  \begin{tcolorbox}[title={#1}, colframe=black!10, coltitle=black]
    \begin{algorithmic}[1]
      \Require #2
      #3
    \end{algorithmic}
  \end{tcolorbox}
  \end{algorithm}
}

\newcommand{\pk}{\textit{pk}}
\newcommand{\sk}{\textit{sk}}
\newcommand{\ak}{\textit{ak}}
\newcommand{\am}{\textit{am}}
\newcommand{\ie}{i.\,e.\xspace}

\newcommand{\eg}{e.\,g.\xspace}
\newcommand{\doclength}{\combine{\chan\,}{\,n}}

\setcopyright{acmlicensed}
\acmDOI{10.1145/3133956.3133981}
\acmISBN{978-1-4503-4946-8/17/10}
\acmConference{CCS '17}{October 30-November 3, 2017}{Dallas, TX,
  USA}
\acmYear{2017}
\copyrightyear{2017}
\acmPrice{15.00}

\begin{document}

\title{Algorithm Substitution Attacks from a\\ Steganographic  Perspective}  
\author{Sebastian Berndt}
\affiliation{%
  \institution{University of L\"{u}beck}
  \streetaddress{Ratzeburger Allee 160}
  \city{L\"{u}beck} 
  \postcode{23562}
  \country{Germany}
}
\email{berndt@tcs.uni-luebeck.de}

\author{Maciej Li\'{s}kiewicz}
\affiliation{%
  \institution{University of L\"{u}beck}
  \streetaddress{Ratzeburger Allee 160}
  \city{L\"{u}beck} 
  \postcode{23562}
  \country{Germany}
}
\email{liskiewi@tcs.uni-luebeck.de}

\begin{abstract}
  The goal of an algorithm substitution attack (ASA), also called a
  subversion attack (SA), is to replace an honest implementation of a
  cryptographic tool by a subverted one which allows to leak private
  information while generating output indistinguishable from the honest
  output.  Bellare, Paterson, and Rogaway provided at CRYPTO '14 a
  formal security model to capture this kind of attacks and constructed
  practically implementable ASAs against a large class of \emph{symmetric
    encryption schemes}. At CCS'15, Ateniese, Magri, and Venturi extended
  this model to allow the attackers to work in a fully-adaptive and
  continuous fashion and proposed subversion attacks against
  \emph{digital signature schemes}. Both papers also showed
  the impossibility of ASAs in cases where the cryptographic tools are
  deterministic. Also at CCS'15, Bellare, Jaeger, and Kane strengthened
  the original model and proposed a universal ASA against sufficiently
  random encryption schemes.  In this paper we analyze ASAs from the
  perspective of steganography -- the well known concept of hiding the
  presence of secret messages in legal communications. While a close
  connection between ASAs and steganography is known, this lacks a
  rigorous treatment. We consider the common computational model for
  secret-key steganography and prove that successful ASAs correspond to
  secure stegosystems on certain channels and vice versa. This formal
  proof allows us to conclude that ASAs are stegosystems and to
  ``rediscover'' several results concerning ASAs known in the
  steganographic literature.
\end{abstract}




\maketitle

\section{Introduction}

The publication of secret internal documents of the NSA by Edward Snowden
(see \eg \cite{ball2013revealed, greenwald2014no,  perlroth2013nsa}) 
allowed the cryptographic community a unique insight into some well-kept
secrets of one of the world's largest security agency.
Two conclusions may be drawn from these reveals:

\begin{itemize}
\item On the one hand, even a large organization such as the NSA seems not to
  be able to break well established implementations of cryptographic
  primitives such as RSA~or~AES\@.
\item On the other hand, the documents clearly show that the NSA
  develops methods and techniques to \emph{circumvent} the well
  established security notions by \eg manipulating standardization
  processes (\eg issues surrounding the number generator \texttt{Dual\_EC\_DRBG} 
  \cite{checkoway2014dual, schneier2007did, shumow2007back})
    or reason about metadata. 
\end{itemize}
This confirms that the security guarantees provided by the cryptographic community 
are sound, but also indicates that some security definitions are too narrow to evade all 
possible attacks, including (non-)intentional improper handling of 
theoretically sound cryptographic protocols. 
A very realistic attack which goes beyond the common framework 
is a modification of an appropriate implementation
of a secure protocol.
The modified implementation should remain indistinguishable from
a truthful one and its aim is to allow leakage of secret information 
during subsequent runs of the subverted protocol.
Attacks of this kind are known in the literature 
\cite{young1997kleptography,young1996dark,bellare2015asa,bellare2014asa,ateniese2015sig,russell2016cliptography}
and an overview 
on this topic 
is given in the current survey
\cite{schneier2015survey} by \citeauthor{schneier2015survey}

A powerful class of such attacks that we will focus on -- coined
\emph{secretly embedded trapdoor with universal
          protection (SETUP) attacks} -- was presented over twenty years ago by
Young and Yung in the \emph{kleptographic} model framework
\cite{young1996dark,young1997kleptography}. The model is
meant to capture a situation where an adversary (or “big brother” as we
shall occasionally say) has the opportunity to implement (and, indeed,
“mis-implement” or subvert) a basic cryptographic tool. 
The difficulty in detecting such an attack is based on the 
hardness of program verification.
By using  \emph{closed source} software, the user 
must trust the developers
that their implementation of cryptographic primitives is truthful and does not
contain any backdoors. This is especially true for hardware-based
cryptography \cite{bellare2014asa}. But it is difficult to verify this property.
Even if the software is \emph{open source} -- the source code is publicly available -- the
sheer complexity of cryptographic implementations allows only  very
specialized experts to be able to judge these
implementations. Two of the most prominent bugs of the widely spread
cryptographic library
\texttt{OpenSSL}\footnote{\url{https://www.openssl.org/}} -- the
\emph{Heartbleed bug} and Debian's faulty implementation of the pseudorandom
number generator -- remained undiscovered for more than two years
\cite{schneier2015survey}.

Inspired by Snowden's reveals, the recent developments  
reignited the interest in these kind of attacks. 
\citeauthor{bellare2014asa} named them \emph{algorithm
substitution attacks (ASA)} and
showed several attacks on certain symmetric encryption schemes
\cite{bellare2014asa}. Note that they defined a very weak 
model, where the only goal of the attacker was to distinguish between
two ciphertexts, but mostly used a stronger scenario with the aim to
recover the encryption key. \citeauthor{degabriele2015asa} criticized the model
of~\cite{bellare2014asa} by pointing out the results 
crucially rely on the fact that a subverted encryption algorithm always needs to
produce valid ciphertexts (the \emph{decryptability assumption}) and
proposed a refined security notion \cite{degabriele2015asa}. The model
of algorithm substitution attacks introduced in 
\cite{bellare2014asa}
was extended to signature schemes
by \citeauthor{ateniese2015sig} in
\cite{ateniese2015sig}. Simultaneously, \citeauthor{bellare2015asa}
\cite{bellare2015asa} 
strengthened the result of \cite{bellare2014asa} by enforcing that the
attack needs to be stateless.

In this paper we thoroughly analyze (general) ASAs from the \emph{steganographic} 
point of view. 
The principle goal of steganography is to hide information in
unsuspicious communication such that no observer can
distinguish between normal documents and documents that carry additional information.
Modern steganography was first made popular due to the prisoners’
problem by Simmons \cite{simmons1984prisoners} but, interestingly,
the model was inspired by detecting the risk of ASAs 
during  development of the SALT2 treaty
between the Soviet Union
and the United States in the late seventies \cite{simmons1998history}.
This sheds some light on the inherent relationship between these two frameworks
which is well known in the literature
(see \eg \cite{young1996dark,young1997kleptography,russell2016destroying}). A related result showing that so called \emph{decoy password
  vaults} are very closely related to stegosystems
on a certain kind of channels was presented by \citeauthor{pasquini2017decoy} in
\cite{pasquini2017decoy}.

Our main achievement is providing a strict relationship between secure 
algorithm substitution attacks and the common computational model for secret-key steganography.
Particularly, we prove
that successful ASAs correspond to secure stegosystems on certain
channels and vice versa. This formal proof allows us to conclude
that ASAs are stegosystems and to ``rediscover'' results of 
\cite{bellare2014asa,bellare2015asa,ateniese2015sig}
concerning ASAs. 

The computational model for steganography used in this paper 
was first presented
by Hopper, Langford, and von Ahn
\cite{hopper2002provably,hopper2009provably} and independently 
proposed 
by Katzenbeisser and Petitcolas \cite{katzenbeisser2002defining}. A
\emph{stegosystem} consists of an encoder and a decoder sharing a
key. The encoder's goal is to \emph{embed} a secret message into a
sequence of documents which are send via a public \emph{communication
  channel} $\chan$ monitored by an adversary (often called the
\emph{warden} due to the prisoners problem of
Simmons~\cite{simmons1984prisoners}). The warden wants to distinguish
documents that carry no secret information from those sent by the
encoder. If all polynomial-time (in the security parameter $\kappa$)
wardens fail to distinguish these cases,
we say that the stegosystem is \emph{secure}. If the decoder is able to
reconstruct the secret message from the sequence send by the encoder,
the system is called \emph{reliable}.

\subsection*{Our Results}
We first investigate 
algorithm substitution attacks against symmetric encryption schemes
in the framework by \citeauthor{bellare2015asa} \cite{bellare2015asa}.
We model encryption schemes as steganographic channels in appropriate way which  
allows to relate algorithm substitution attacks with steganographic systems and vice versa.
This leads to the following result. 

\begin{theorem}[Informal]\label{thm:stego:asa:iff:on:ses}
Assume that $\SES$ is a symmetric encryption scheme. Then  
there exists an indistinguishable and reliable algorithm substitution
attack against $\SES$ if and only if there exists a secure and reliable 
stegosystem on the channel determined by $\SES$. 
\end{theorem}
The proof of the theorem is constructive in the sense that  
we give an explicit construction of an algorithm substitution
attack against $\SES$ from a stegosystem and vice versa. 
As conclusion we provide a generic ASA against \emph{every} symmetric 
encryption scheme $\SES$ whose insecurity is negligible if, roughly speaking,  
$\SES$ has  sufficiently large min-entropy.
Our algorithm against $\SES$ achieves almost the same 
performance as the construction of \citeauthor{bellare2015asa} 
(see Theorem~4.1 and Theorem~4.2 in \cite{bellare2015asa} and also our discussion 
in Section~\ref{Sec:ASA:against:encrypted:as:stego}).

Next, we generalize our construction and show a generic algorithm 
substitution attack $\ASAE$ against any (polynomial-time) randomized algorithm $\R$
which, with hardwired secret $s$, takes inputs $x$ and generates outputs $y$.
Algorithm $\ASAE$, using a hidden hardwired random key $\ak$,
returns upon the secret $s$ the sequence $\tilde{y}_1,\tilde{y}_2,\ldots$
such that the output is indistinguishable from
$\R(s,x_1), \R(s,x_2),\ldots$  and  $\tilde{y}_1,\tilde{y}_2,\ldots$ 
embeds the secret $s$. From this result we  conclude:

\begin{theorem}[Informal]
  \label{thm:generic-attack}
  There exists a generic algorithm 
  substitution attack  $\ASA$ that allows an 
  undetectable subversion of any cryptographic primitive of
  sufficiently large min-entropy. 
\end{theorem}

\begin{theorem}[Informal]
  \label{thm:min-entropy}
  Let $\Pi$ be a cryptographic primitive consisting with algorithms
  $(\Pi.A_{1},\Pi.A_{2},\ldots,$ $ \Pi.A_{r})$ such that $\{A_{i}\mid i\in
  I\}$ for some $I\subseteq \{1,\ldots,r\}$ are deterministic. 
  Then there is no ASA on $\Pi$ which subverts only algorithms 
  $\{A_{i}\mid i\in I\}$.
\end{theorem}

As a corollary we obtain the result of \citeauthor{ateniese2015sig} 
(Theorem~1 in \cite{ateniese2015sig}) that for every \emph{coin-injective} signature 
scheme, there is a successful algorithm substitution attack  
of negligible insecurity. Moreover we get (Theorem~2 in \cite{ateniese2015sig}) 
that for every \emph{coin-extractable} signature scheme, there is a successful
and secure ASA\@. We can conclude also (Theorem~3 in \cite{ateniese2015sig}) 
that \emph{unique signature schemes} are resistant to ASAs fulfilling the \emph{verifiability condition}. 
Roughly speaking the last property means that each message has exactly
one signature and the ASA can only produce valid signatures.

We furthermore introduce the concept of \emph{universal ASAs} that can
be used without a detailed description of the implementation of the
underlying cryptographic primitive and note that almost all known ASAs
belong to this class. Based upon this definition, we prove the following
upper bound on the information that can be embedded into a single
ciphertext:
\begin{theorem}[Informal]
  \label{thm:universal:bound}
  No universal ASA is able to embed more than
  $\mathcal{O}(1)\cdot \log(\kappa)$ bits of information into a single
  ciphertext.
\end{theorem}

The paper is organized as follows. Section~\ref{sec:prelim} contains the
basic preliminaries and notations that we use throughout this work,
Section~\ref{sec:asa} presents the formal definitions of algorithm
substitution attacks, and Section~\ref{sec:stego} gives the necessary
background on steganography. In order to relate ASAs and steganography,
we make use of an appropriate channel for symmetric encryption schemes
defined in Section~\ref{sec:channel}. The proof of 
Theorem~\ref{thm:stego:asa:iff:on:ses} is given in
Section~\ref{Sec:ASA:against:encrypted:as:stego}, where one direction is
contained in Theorem~\ref{thm:asa:on:ses:impl:stego} and the other
direction is given as Theorem~\ref{thm:ses:on:ses:impl:stego}. We
generalize our results to arbitrary randomized algorithms in
Section~\ref{sec:general}. Combining the positive results of
Theorem~\ref{thm:generic-attack:against:R} with the generic stegosystem
provided by Theorem~\ref{thm:rejsam:secure} allows us to conclude
Theorem~\ref{thm:generic-attack}. The negative results of
Theorem~\ref{thm:no-attack:against:R} directly give
Theorem~\ref{thm:min-entropy}. Finally, Section~\ref{sec:bound} defines
universal ASAs and contains the upper bound on the transmission rate of
these ASAs via a sequence of lemmata that results in
Corollary~\ref{cor:bound} implying Theorem~\ref{thm:universal:bound}.

\section{Basic Preliminaries and Notations}
\label{sec:prelim}
We use the following standard notations.
A function $f\colon \nats \to \nats$ is \emph{negligible}, if for all
$c\in \nats$, there is an $n_{0}\in \nats$ such that $f(n) < n^{-c}$ for
all $n\geq n_{0}$. The set of all strings of length $n$ on an alphabet
$\Sigma$ is denoted by $\Sigma^{n}$ and the set of all strings of length
at most $n$ is denoted by $\Sigma^{\leq
  n}:=\cup_{i=0}^{n}\Sigma^{i}$. If $S$ is a set, $x\gets S$ denotes the
uniform random assignment of an element of $S$ to $x$. If $\mathsf{A}$ is
a randomized algorithm, $x\gets \mathsf{A}$ denotes the random
assignment (with regard to the internal randomness of $\mathsf{A}$) 
of the output of $\mathsf{A}$ to $x$. 
The \emph{min-entropy} measures the amount of randomness of a probability
distribution $D$ and is defined as
$\minent(D)=\inf_{x\in \supp(D)}\{-\log \Pr_{D}(x)\}$, where $\supp(D)$
is the support of $D$. 
Moreover, PPTM stands for probabilistic polynomial-time Turing machine. 

A \emph{symmetric encryption scheme} $\SES$ is a triple of probabilistic
polynomial-time algorithms 
$(\combine{\SES}{\Gen},\combine{\SES}{\Enc},\combine{\SES}{\Dec})$
with parameters $\combine{\SES}{\ml}(\kappa)$ describing the length of the encrypted message
and  $\combine{\SES}{\cl}(\kappa)$ describing the length of a generated
cipher message. The  algorithms have the following properties:
\begin{itemize}
\item The \emph{key generator} $\combine{\SES}{\Gen}$ produces upon input $1^{\kappa}$ a
  key $k$ with $|k|=\kappa$.
\item The \emph{encryption algorithm} $\combine{\SES}{\Enc}$ takes as input the key $k$
  and a message $m\in \{0,1\}^{\combine{\SES}{\ml}(\kappa)}$ of length
  $\combine{\SES}{\ml}(\kappa)$ and produces a \emph{ciphertext}
  $c\in \{0,1\}^{\combine{\SES}{\cl}(\kappa)}$ of length
  $\combine{\SES}{\cl}(\kappa)$. 
\item The \emph{decryption algorithm} $\combine{\SES}{\Dec}$ takes as input the key $k$
  and a ciphertext $c\in \{0,1\}^{\combine{\SES}{\cl}(\kappa)}$ and produces a message $m'\in
  \{0,1\}^{\combine{\SES}{\ml}(\kappa)}$. 
\end{itemize}
If the context is clear, we also write $\Gen$, $\Enc$, $\Dec$, $\ml$ and
$\cl$ without the prefix $\SES$. We say that $(\Gen,\Enc,\Dec)$ is \emph{reliable}, if
$\Dec(k,\Enc(k,m))=m$ for all $k$ and all $m$.

An \emph{cpa-attacker} $\Att$ against a symmetric encryption scheme 
is a PPTM that
mounts \emph{chosen-plaintext-attacks (cpa)}: It is given a challenging
oracle $\CH$ that either equals $\Enc_{k}$ for a randomly generated key
$k$ or produces random bitstrings of length $\cl(\kappa)$. For an
integer $\lambda$, let $\RAND(\lambda)$ be an algorithm that returns
uniformly distributed bitstrings of length $\lambda$.  The goal of
$\Att$ is to distinguish between those settings. Formally, this is
defined via the following experiment named $\CPA-Dist$:

\mlAlgorithm{$\CPA-Dist_{\Att,\SES}(\kappa)$}
{attacker $\Att$, 
  symmetric encryption scheme  $\SES=(\Gen,\Enc,\Dec)$}
  {
    \State $k\gets \Gen(1^{\kappa})$; $b\gets \{0,1\}$
    \State $b' \gets \Att^{\CH}(1^{\kappa})$
    \State \AlgReturn{$b=b'$}

    \Statex
\setcounter{ALG@line}{0}
\Statex \hspace{-.4cm}\underline{oracle $\CH(m)$ }
\State {\bf if} $b=0$ {\bf then}  \AlgReturn{$\Enc(k,m)$}\newline
\hspace*{10mm}{\bf else} \AlgReturn{$\RAND(\cl(\kappa))$}
  }
  {Chosen-Plaintext-Attack experiment with security parameter $\kappa$.}

A symmetric encryption scheme $\SES$ is \emph{cpa-secure}
if for every attacker $\Att$ there is a negligible function $\negl$
such that
\begin{align*}
  &\adv^{\cpa}_{\SES}(\kappa) :=
  |\Pr[\CPA-Dist_{\Att,\SES}(\kappa)=\textsf{true}] - 1/2|\leq
  \negl(\kappa).
\end{align*}
The maximal advantage of any attacker against $\SES$ is
called the \emph{insecurity} of $\SES$ and is defined as
\begin{align*}
  \insec^{\cpa}_{\SES}(\kappa) := \max_{\Att}\{\adv_{\Att,\SES}^{\cpa}(\kappa)\}.
\end{align*}

For a $\SES=(\Gen,\Enc,\Dec)$
we will assume that it has nontrivial randomization
measured by the min-entropy $\minent(\SES)$ of ciphertexts  
that is defined via
\[
    2^{-\minent(\SES)} = \max_{k,m,c} \Pr[\combine{\SES}{\Enc}(k,m)=c].
\]

For two numbers $\ell,\ell'\in \nats$, denote the \emph{set of all function} from
$\{0,1\}^{\ell}$ to $\{0,1\}^{\ell'}$ by $\Fun(\ell,\ell')$. Clearly, in order to
specify a random element of $\Fun(\ell,\ell')$, one needs $2^{\ell}\times \ell'$
bits and we can thus not use completely random functions in an efficient
setting. Therefore we will use efficient functions that are indistinguishable
from completely random functions. A \emph{pseudorandom function} is a
pair of PPTMs
$\algf=(\combine{\algf}{\Eval},\combine{\algf}{\Gen})$  such that $\combine{\algf}{\Gen}$ upon input $1^{\kappa}$ produces
a key $k\in \{0,1\}^{\kappa}$. The keyed function $\combine{\algf}{\Eval}$
takes the key $k\gets \combine{\algf}{\Gen}(1^{\kappa})$ and a bitstring $x$ of
length $\combine{\algf}{\fin}(\kappa)$ and produces a string $\combine{\algf}{\Eval}_{k}(x)$ of length
$\combine{\algf}{\fout}(\kappa)$. An attacker, called \emph{distinguisher} $\Dist$,
is a PPTM that upon
input $1^{\kappa}$ gets
oracle access to a function that either equals
$\combine{\algf}{\Eval}_{k}$ for a randomly chosen key $k$ 
or is a completely random function $f$.
The goal of $\Dist$ is to distinguish
between those cases. A pseudorandom function $\algf$ is secure if for every 
distinguisher $\Dist$ there is a
negligible function $\negl$ such that
\begin{align*}
  &\Adv_{\Dist,\algf}^{\prf}(\kappa) \ := \\ 
  & \quad \quad \left| \Pr[\Dist^{\combine{\algf}{\Eval}_{k}}(1^{\kappa})=1] -
  \Pr[\Dist^{f}(1^{\kappa})=1]\right| \leq \negl(\kappa),
\end{align*}
where $k\gets \combine{\algf}{\Gen}(1^{\kappa})$ and $f\gets
\Fun(\combine{\algf}{\fin}(\kappa),\combine{\algf}{\fout}(\kappa))$. If
$\Dist$ outputs $1$, this means that the distinguisher $\Dist$ believes
that he deals with a truly random function. 

As usual, the maximal advantage of any distinguisher against
$\algf$ is called the \emph{prf-insecurity}
$\InSec_{\algf}^{\prf}(\kappa)$ and defined as
\begin{align*}
  \InSec_{\algf}^{\prf}(\kappa) := \max_{\Dist}\{\Adv^{\prf}_{\Dist,\algf}(\kappa)\}.
\end{align*}

\section{Algorithm Substitution Attacks against Encryption Schemes}
\label{sec:asa}
While it is certainly very useful for an attacker to be able to
reconstruct the key, one can also consider situations, where the
extractor should be able to extract different information from the
ciphertexts or signatures. We will thus generalize the algorithm
substitution attacks described in the literature to the setting, 
where the substituted algorithm also takes a message $\am$ as argument 
and the goal of the extractor is to derive this message from 
the produced ciphertext. By always setting $\am := k$, this is the
setting described by \citeauthor{bellare2015asa} in
\cite{bellare2015asa}.
We thus strengthen the model of \cite{bellare2014asa} and \cite{bellare2015asa}
in this sense.  

Below we give in detail our definitions based upon
the model proposed by Bellare et al.~in \cite{bellare2015asa}.
If the substitution attack is \emph{stateful}, we allow 
the distinguisher that tries to identify the attack to also choose this
state and observe the internal state of the attack. Every algorithm
substitution attack thus needs to be \emph{stateless}, as in the model of
Bellare et al.~in \cite{bellare2015asa}. Note that this is a stronger
requirement than in \cite{bellare2014asa} and
\cite{ateniese2015sig}, as those works also allowed stateful attacks.

In our setting an \emph{algorithm substitution attack} 
against a symmetric encryption scheme
$\SES=(\combine{\SES}{\Gen},\combine{\SES}{\Enc},\combine{\SES}{\Dec})$ 
is a triple of PPTMs 
\[
  \begin{array}{rcl}
  \ASAE&=&(\combine{\ASAE}{\Gen},\combine{\ASAE}{\Enc},\combine{\ASAE}{\Ext}) \\
 \end{array}
\]
with parameter $\combine{\ASAE}{\ml}(\kappa)$ for the \emph{message
  length} -- the length of the attacker message  --
and the following functionality.
\begin{itemize}
\item 
  The \emph{key generator} $\combine{\ASAE}{\Gen}$ produces upon input $1^{\kappa}$
  an attacker key $\ak$ of length $\kappa$.
\item 
  The \emph{encryption algorithm} $\combine{\ASAE}{\Enc}$
  takes 
  an attacker key $\ak\in \supp(\combine{\ASAE}{\Gen}(1^{\kappa}))$, 
   attacker message $\am$ such that $\am \in \{0,1\}^{\combine{\ASAE}{\ml}(\kappa)}$, 
  an encryption key $k\in \supp(\combine{\SES}{\Gen}$ $(1^{\kappa}))$, an
  encryption message $m\in \{0,1\}^{\combine{\SES}{\ml}(\kappa)}$, and a
  state $\sigma\in \{0,1\}^{*}$ 
  and produces a ciphertext
  $c$ of length $\combine{\SES}{\cl}(\kappa)$ and a new state~$\sigma'$.
\item 
  The \emph{extraction algorithm} $\combine{\ASAE}{\Ext}$ takes as input an attacker  key
  $\ak\in \supp(\combine{\ASAE}{\Gen}(1^{\kappa}))$ and $\ell =
                                \combine{\ASAE}{\outl}(\kappa)$
                                a ciphertext
   $c_{1},\ldots,c_{\ell}$ with 
  $c_{i}\in
 \{0,1\}^{\combine{\SES}{\cl}(\kappa)}$ and produces an
 attacker message
 $\am'$. 
\end{itemize}

An algorithm substitution attack needs (a) to be indistinguishable from
the symmetric encryption scheme and (b) should be able to reliably extract the
message $\am$ of length $\combine{\ASAE}{\ml}(\kappa)$ from the
ciphertexts.
Due to information-theoretic
reasons, it might be impossible to embed the attacker message $\am$ into
a single ciphertext: If $\combine{\SES}{\Enc}$ uses $10$ bits of
randomness, at most $10$ bits from $\am$ can be reliably embedded into a
ciphertext. Hence, the algorithm substitution attack needs to produce
more than one ciphertext in this case. For
message $m_{1},\ldots,m_{\ell}$, the
complete output, denoted as
$\combine{\ASAE}{\Enc}^{\ell}(\ak,\am,k,m_{1},\ldots,m_{\ell})$
is
defined as follows:
\begin{enumerate}
\item[1:] $\sigma = \varnothing$
\item[2:] {\bf for} $j=1$ to $\ell$ {\bf do} $(c_{j},\sigma)\gets \combine{\ASAE}{\Enc}(\ak,\am,k,m_{j},\sigma)$
\item[3:] {\bf return} $c_{1},\ldots,c_{\ell}$
\end{enumerate}

  To formally define the probability that the extractor is able to
  reliably extract $\am$ from the given ciphertexts
  $c_{1},\ldots,c_{\ell}$, we define its \emph{reliability}\footnote{In
    \cite{bellare2015asa}, this is called the \emph{key recovery
      security}.} as $1-\unrel_{\ASAE,\SES}(\kappa)$, where the
  \emph{unreliability} $\unrel_{\ASAE,\SES}$ is given as
\begin{align*}
  &\max                                     \{
\Pr[\combine{\ASAE}{\Ext}(\ak,\combine{\ASAE}{\Enc}^{\ell}(\ak,\am,k,m_{1},\ldots,m_{\ell}))\neq
  \am]\},
\end{align*}
with  the maximum  taken over all
$\ak\in \supp(\combine{\ASAE}{\Gen}(1^{\kappa})), 
 \am \in\{0,1\}^{\combine{\ASAE}{\ml}(\kappa)}$, 
and   $m_{i}\in\{0,1\}^{\combine{\SES}{\ml}(\kappa)}$.
The algorithm 
is \emph{successful}, if there is
negligible function $\negl$ with $\unrel_{\ASAE,\SES}(\kappa) \leq \negl(\kappa).$

The indistinguishability of an ASA is defined as follows.
Call a \emph{watchdog} $\Watch$ a PPTM 
that tries to distinguish the output of the attacker encryption
algorithm $\combine{\ASA}{\Enc}$ from the original encryption algorithm
$\Enc$.
The indistinguishability  is defined via the 
game named $\EncASA-Dist$:

\mlAlgorithm{$\EncASA-Dist_{\Watch,\ASAE,\SES}(\kappa)$}{watchdog $\Watch$, algorithm substitution attack
  $\ASAE=(\combine{\ASAE}{\Gen}, \combine{\ASAE}{\Enc},$ $ \combine{\ASAE}{\Ext})$, and encryption
  scheme $\SES=(\combine{\SES}{\Gen},\combine{\SES}{\Enc},\combine{\SES}{\Dec})$}
{
\State $\ak\gets \combine{\ASAE}{\Gen}(1^{\kappa})$; $b\gets \{0,1\}$
\State $b' \gets \Watch^{\CH}(1^{\kappa})$
\State \AlgReturn{$b=b'$}
\Statex
\setcounter{ALG@line}{0}
\Statex \hspace{-.4cm}\underline{oracle $\CH(\am,k,m,\sigma)$ }
\State {\bf if} $b=0$ {\bf then} $c\gets \combine{\SES}{\Enc}(k,m)$\newline
\hspace*{10mm}{\bf else}  $(c,\sigma) \gets \combine{\ASAE}{\Enc}(\ak,\am,k,m,\sigma)$
\State \Return{$(c,\sigma)$}

}{ASA-distinguishing (detection) experiment with security parameter $\kappa$.}

An algorithm substitution attack $\ASAE$ is called \emph{indistinguishable} from
the symmetric encryption scheme $\SES$, 
if for every watchdog $\Watch$, there is a negligible function $\negl$
such that
\begin{align*}
  &\adv^{\encwatch}_{\Watch, \ASAE,\SES}(\kappa) \
   :=\\ &\quad \quad |\Pr[\EncASA-Dist_{\Watch,\ASAE,\SES}(\kappa)=\textsf{true}] -1/2| 
    \leq \ \negl(\kappa).
\end{align*}

The maximal advantage of any watchdog distinguishing
$\ASAE$ from $\SES$ is
called the \emph{indistinguishability} or \emph{insecurity} of $\ASAE$ and is defined as
\begin{align*}
  \insec^{\encwatch}_{\ASAE,\SES}(\kappa) := \max_{\Watch}\{\adv_{\Watch,\ASAE,\SES}^{\encwatch}(\kappa)\}.
\end{align*}

In \cite{bellare2014asa}, Bellare et al.~proposed a (stateless) construction $\ASAE$
against all symmetric encryption schemes $\SES$. 
They prove in Theorem~3 that if $\SES$ is a randomized, stateless, coin-injective symmetric encryption scheme
with randomness-length $r$ and  if the ASA uses a PRF
$\algf$
then for a watchdog  $\Watch$ that makes $q$ queries to its $\CH$ oracle
we can construct an adversary $\Att$ such that 
$ \adv^{\encwatch}_{\Watch, \ASAE,\SES}(\kappa) \le q/2^{2^r} +
\adv^{\prf}_{\Att,\algf}(\kappa)$, where $\Att$ makes $q$ oracle queries
and its running time is that of $\Watch$. 

Bellare et al.~conclude that as long as their scheme uses a non-trivial amount 
of randomness, for example $r \ge 7$ bits resulting $2^r \ge 128$, 
Theorem 3 implies that the subversion is undetectable.

\section{Backgrounds of Steganography}
\label{sec:stego}
The definitions of the basic steganography concepts 
presented in this section are essentially 
those of \cite{hopper2009provably} and \cite{dedic2009upper}.

In order to define undetectable hidden communication, we need to
introduce a notion of \emph{unsuspicious} communication. We do this via
the notion of a \emph{channel} $\chan$. A channel $\chan$ on the alphabet
$\Sigma$ with maximal document length $\doclength$ is a function that maps a string of previously
send elements $h\in (\Sigma^{\leq \doclength})^{*}$ -- the \emph{history} -- to a
probability distribution upon $\Sigma^{\leq \doclength}$. We denote
this probability distribution by $\chan_{h}$. The elements of
$\Sigma^{\leq \doclength}$ are called \emph{documents}. 
As usually, we will assume that the sequences of documents 
are efficiently prefix-free  recognizable.

A \emph{stegosystem} $\Steg$ on a family of channels 
$\bm{\chan}=\{\chan^{\kappa}\}_{\kappa\in \nats}$ is a triple of
probabilistic polynomial-time (according to the security parameter $\kappa$) algorithms:
\[
\begin{array}{rcl}
  \Steg&=&(\combine{\Steg}{\Gen},\combine{\Steg}{\Enc},\combine{\Steg}{\Dec}) 
\end{array}
\]
with parameters $\combine{\Steg}{\ml}(\kappa)$ describing the \emph{message
length} of the subliminal (hidden, or attacker) message
and  $\combine{\Steg}{\outl}(\kappa)$ describing the length of a generated sequence 
of stego documents to embed the whole hidden message. The algorithms
have the following functionality:
\begin{itemize}
\item The \emph{key generator} $\combine{\Steg}{\Gen}$ takes the unary presentation 
of an integer $\kappa$ -- the \emph{security parameter} -- and outputs 
a key (we will call it an attacker key) 
$\ak\in \{0,1\}^{\kappa}$ of length $\kappa$.
\item The \emph{stegoencoder} 
$\combine{\Steg}{\Enc}$ takes as input the key $\ak$,  the attacker (or
hidden) message $\am\in
\{0,1\}^{\combine{\Steg}{\ml}(\kappa)}$, a history $h$, and a state $\sigma$ and outputs a 
document $d$
from $\chan^{\kappa}$ such 
that $\am$ is (partially) embedded in this document and a new state. In order to produce
the document, $\combine{\Steg}{\Enc}$ also has sampling access to
$\chan^{\kappa}_{h}$.  We denote this by writing $\combine{\Steg}{\Enc}^{\chan}(\ak,\am,h,\sigma)$.
\item The \emph{(history-ignorant) stegodecoder} $\combine{\Steg}{\Dec}$ takes as input the key
$\ak$ and $\ell=\combine{\Steg}{\outl}(\kappa)$ documents $d_{1},\ldots,d_{\ell}$ and
outputs a message $\am'$. A history-ignorant
stegodecoder thus has no knowledge of previously sent 
documents. The stego\-decoders of nearly all known systems are
history-ignorant.
\end{itemize}

To improve readability, if the
stegosystem is clear from the context, we will omit the prefix
$\Steg$. If $\bm{\chan}=\{\chan^{\kappa}\}_{\kappa\in \mathbb{N}}$ is a
family of channels, the 
\emph{min-entropy}  of $\minent(\bm{\chan},\kappa)$ 
is defined as
$\minent(\bm{\chan},\kappa) = \min_{h\in
  \Sigma^{*}}\{\minent(\chan^{\kappa}_{h})\}$. 
In order to be useful, the stegodecoder should \emph{reliably} decode
the embedded message from the sequence of documents. As in the setting
of algorithm substitution attack, the complete output  of $\ell$~documents
of the stegosystem for the history $h$ 
on the subliminal message $\am$ of length
$\combine{\Steg}{\ml}(\kappa)$ is
denoted as $\combine{\Steg}{\Enc}^{\ell,\chan}(\ak,\am,h)$ and is defined as follows.
\begin{enumerate}
\item[1:] $\sigma = \varnothing$
\item[2:] {\bf for} $j=1$ to $\ell$ {\bf do}
\item[3:]\hspace*{10mm}$(d_{j},\sigma)\gets\combine{\Steg}{\Enc}^{\chan}(\ak,\am,h,\sigma)$;
  \quad $h = h \mid\mid d_{j}$
\item[4:] {\bf return} $d_{1},\ldots,d_{\ell}$
\end{enumerate}

The
\emph{unreliability} $\unrel_{\Steg,\chan}(\kappa)$ of the stegosystem
$\Steg$ on the channel family $\{\chan^{\kappa}\}_{\kappa\in \nats}$ with security parameter
$\kappa$ is defined as
\begin{align*}
 & \unrel_{\Steg,\chan}(\kappa) := \\
 &  \max_{\ak, \am} \max_{h}
 \{\Pr[\combine{\Steg}{\Dec}(\ak,\combine{\Steg}{\Enc^{\ell,\chan}}(\ak,\am,h))\neq \am]\},
\end{align*}
where the maximum is taken over all 
$\ak\in \supp(\combine{\Steg}{\Gen}(1^{\kappa})), 
 \am \in \{0,1\}^{\combine{\Steg}{\ml}(\kappa)}$, 
and  $h\in(\Sigma^{n(\kappa)})^{*}$. 
If there is a negligible function $\negl$ such that
$\unrel_{\Steg,\chan}(\kappa)\leq \negl(\kappa)$, we say that $\Steg$ is
\emph{reliable} on $\chan$.
 Furthermore,  the \emph{reboot-reliability} of the stegosystem $\Steg$ is
defined as
\begin{align*}
& \unrel^{\star}_{\Steg,\chan}(\kappa) := \\
&\max_{\ak, \am} \max_{\tau}\max_{ h_1,\ldots,h_{\tau}}\max_{\ell_1,\ldots,\ell_{\tau}} \{
  \Pr[\combine{\Steg}{\Dec}(\ak,
d_1,d_2,\ldots,d_{\ell}
  )\neq \am]\}  
\end{align*}
where the maxima are taken over all 
$\ak\in \supp(\combine{\Steg}{\Gen}(1^{\kappa}))$, 
 $\am \in \{0,1\}^{\combine{\Steg}{\ml}(\kappa)}$, 
all positive integers $\tau\le \ell$, all histories
$h_{1},\ldots,h_{\tau}$, and all positive integers $\ell_1,\ldots,\ell_{\tau}$
such that $\ell_1 + \ldots + \ell_{\tau}=\ell$. 
The documents $d_{1},\ldots, d_{\ell}$ are the
concatenated output of the runs
\begin{align*}
 \combine{\Steg}{\Enc^{\ell_1,\chan}}(\ak,\am,h_1)\mid\mid
    \ldots \mid\mid
  \combine{\Steg}{\Enc^{\ell_\tau,\chan}}(\ak,\am,h_\tau).  
\end{align*}
We say that the stegosystem $\Steg$ is
\emph{reboot-reliable}
if $\unrel^{\star}_{\Steg,\chan}(\kappa)$ 
is bounded from above by a negligible function.
This corresponds to a situation where the stegoencoder 
is  restarted $\tau$ times, each time with the history $h_i$,
and is allowed to generate $\ell_i$ documents. Note that
reboot-reliability is a strictly stronger requirement than reliability
and we
can thus conclude
\begin{align*}
  \unrel_{\Steg,\chan}(\kappa)\leq \unrel^{\star}_{\Steg,\chan}(\kappa).
\end{align*}

To define the security of a stegosystem, we first specify the abilities
of an attacker: A \emph{warden} $\Ward$ is a probabilistic polynomial-time
algorithm that will have access to a \emph{challenge oracle}
$\CH$. This challenge oracle can be called with a message
$\am$ and a history $h$ and is either equal to 
$\combine{\Steg}{\Enc}^{\chan}(\ak,\am,h,\sigma)$ for a
key $\ak\gets \combine{\Steg}{\Gen}(1^{\kappa})$ or equal to random documents of the \emph{channel}.

The goal of the warden is to distinguish between those oracles. It also
has access to samples of the channel $\chan^{\kappa}_{h}$ for a freely
chosen history $h$. Formally, the
\emph{chosen-hiddentext-attack-advantage} is defined via the following
game $\SS-CHA-Dist$:

\mlAlgorithm{$\SS-CHA-Dist_{\Ward,\Steg,\chan}(\kappa)$}{warden $\Ward$, stegosystem
  $\Steg$,
 channel $\chan$}
{
  \State $\ak\gets \combine{\Steg}{\Gen}(1^{\kappa})$
  \State $b\gets \{0,1\}$
\State $b' \gets \Ward^{\CH,\chan}(1^{\kappa})$
\State \AlgReturn{$b=b'$}
\Statex
\setcounter{ALG@line}{0}
\Statex \hspace{-.4cm}\underline{oracle $\CH(\am,h,\sigma)$ }
\State {\bf if} $b=0$ {\bf then} $d \gets \chan^{\kappa}_{h}$\newline
\hspace*{10mm}{\bf else} $(d,\sigma)\gets \combine{\Steg}{\Enc}(\ak,\am,h,\sigma)$
\State \Return{$(d,\sigma)$}
}{Chosen-Hiddentext experiment with security parameter $\kappa$.}

A stegosystem $\Steg$  is called \emph{secure against chosen-hiddentext attacks} 
if for every warden $\Ward$, there is a negligible function $\negl$
such that
\begin{align*}
  \adv^{\cha}_{\Ward, \Steg,\chan}(\kappa) 
   &:= |\Pr[\SS-CHA-Dist_{\Ward,\Steg,\chan}(\kappa)=\textsf{true}] -1/2| \\
  & \leq \ \negl(\kappa).
\end{align*}

The maximal advantage of any warden against $\Steg$ is the \emph{insecurity}
$\insec^{\cha}_{\Steg,\chan}(\kappa)$ 
and defined as
$\max_{\Ward}\{\adv^{\cha}_{\Ward,\Steg,\chan}(\kappa)\}$.

A very common technique in the design of secure stegosystems called
\emph{rejection sampling} goes back to an idea of
Anderson,  presented in
\cite{anderson1996limits}. The basic concept is that the stegoencoder samples from the
channel until he finds a document that already encodes the
hiddentext. This was first used by Cachin in
\cite{Cachin2004} to
construct a secure stegosystem in the information-theoretic sense.

In the following, let $\algf$ be pseudorandom function that maps input strings
of length $\combine{\algf}{\fin}(\kappa)$ (documents) to strings of
length $\combine{\algf}{\fout}(\kappa)=\log(\ml(\kappa))+1$ (message
parts). To simplify notation, we treat the output of
$\combine{\algf}{\Eval}_{k}$ as a pair $(b,j)$ with $|b|=1$ and
$|j|=\log(\ml(\kappa))$. 
The encoder of the \emph{rejection sampling stegosystem}, 
which we denote as  $\RejSam^{\algf}$, 
is defined as follows:

\inputAlgorithm{$\combine{\RejSam^{\algf}}{\Enc}(\ak,\am,h,\sigma)$}{key
  $\ak$, message $\am$, history $h$, state $\sigma$}{
  \State $i := 0$;
\Repeat
\State $d \gets \chan_{h}$ 
\State $i := i+1$ 
\State $(b,j) := \combine{\algf}{\Eval}_{\ak}(d)$ 
\Until{$\am[j]=b$ or $i > s$}\Comment{$\am[j]$ is the $j$-th
    bit of $\am$}
\State \AlgReturn{$(d,\sigma)$}
}{Stegoencoder of $\RejSam$ with security parameter $\kappa$ and $s\ge 1$.}

The key generator $\combine{\RejSam^{\algf}}{\Gen}$ is equal to
$\combine{\algf}{\Gen}$ and the decoder derives $\am$, as long as its
input documents contain every bit $\am[j]$, by applying
$\combine{\algf}{\Eval}_{\ak}$ to these documents. 
Below we present the description of the decoder. 
Note that the stegosystem is stateless. 

\inputAlgorithm{$\combine{\RejSam^{\algf}}{\Dec}(\ak,
  d_{1},\ldots,d_{\combine{\Steg}{\outl}(\kappa)})$}{key $\ak$,
  documents $d_{1},\ldots,d_{\combine{\Steg}{\outl}(\kappa)}$ 
  }{
  \For{$j=1,\ldots,\ml(\kappa)$}
  \State let $\am_{j} := \bot$ 
  \EndFor
  \For{$i=1,2,\ldots,\combine{\Steg}{\outl}(\kappa)$}
  \State $(b,j) := \combine{\algf}{\Eval}_{k}(d_{i})$
  \State let $\am_{j} := b$
  \EndFor
  \If{all $\am_{j}\neq \bot$}
  \State\AlgReturn{$\am=\am_{1}\am_{2}\ldots \am_{\ml(\kappa)}$}
  \Else
  \State\AlgReturn{$\bot$}
  \EndIf
}{Decoder of $\RejSam$.}

In \cite{hopper2009provably}, Hopper et al.~were the first to prove the
security of this stegosystem in the complexity-theoretic model. Their
argument was simplified by Dedi{\'c} et al.~in \cite{dedic2009upper} and
by Backes and Cachin in \cite{backes2005public}. The version given here
is based upon the stateless construction of Dedi{\'c} et al.~and also
uses the idea of Bellare et al.~in \cite{bellare2015asa} to apply the
\emph{coupon collector's problem} to completely get rid of the state by
randomly choosing an index to embed.

The analysis of the coupon collector's problem shows that by sending
$\ml(\kappa)\cdot (\ln \ml(\kappa)+\beta)$ documents~--~for an appropriate
value $\beta$~--~one only introduces a term
$\exp(-\beta)$ into the unreliability (see \eg
\cite{mitzenmacher2005probability} for a proof of this fact), which can
be made negligible by setting $\beta\geq \ml(\kappa)-\ln(\ml(\kappa))$. The output length
on messages of length $\ml(\kappa)$
will thus be bounded by $\ml(\kappa)^{2}$.

The security of this system directly follows from the analysis of
Dedi{\'c} et al.~in \cite{dedic2009upper}:

\begin{theorem}[{\cite[Theorems 4 and 5]{dedic2009upper}}]
\label{thm:rejsam:secure}
  For every polynomial $\ml(\kappa)$, there exists a universal history-ignorant
  stego\-system $\Steg=\RejSam^{\algf}$ with security parameter $\kappa$ and $s\ge 1$
  such that for every channel $\chan^{\kappa}$ we have
  \begin{itemize}
  \item $\combine{\Steg}{\ml}(\kappa)=\ml(\kappa)$,
  \item $\InSec^{\cha}_{\Steg,\chan}(\kappa) \leq
    \mathcal{O}(\ml(\kappa)^{4}\cdot
    2^{-\minent(\chan^{\kappa})}+\ml(\kappa)^{2}\cdot
    \exp(-s))+\InSec^{\prf}_{\algf,\chan}(\kappa)$, and
  \item $\unrel^{\star}_{\Steg,\chan}(\kappa) \leq
    \ml(\kappa)^{2}(2\cdot
    \exp(-2^{\minent(\chan^{\kappa})-3}) +
    \exp(-2^{-2} s))+
    \InSec^{\prf}_{\algf,\chan}(\kappa)$.
  \end{itemize}
\end{theorem}

The notation $\InSec^{\prf}_{\algf,\chan}(\kappa)$ indicates the
insecurity of the pseudorandom function $\algf$ \emph{relative to the
  channel $\chan$}. Informally, this means that the attacker against
$\algf$ also has sampling access to $\chan$ (for a formal definition,
see \cite{dedic2009upper}). For an \emph{efficiently sampleable} channel
$\chan$ (\ie one that can be simulated by a PPTM), it clearly holds that
$\InSec^{\prf}_{\algf,\chan}(\kappa)=\InSec^{\prf}_{\algf}(\kappa)$. 
All channels used in this work are efficiently sampleable and we will thus
omit 
the index $\chan$ from the term $\InSec$.

\section{Encryption Schemes as Steganographic Channels}
\label{sec:channel}
Let $\SES=(\Gen,\Enc,\Dec)$ be a symmetric encryption scheme  that encodes
messages of length $\ml(\kappa)$ into ciphertexts of length
$\cl(\kappa) \geq \ml(\kappa)$ and let $\ell$ be a polynomial of $\kappa$.
For $\SES$ we define a channel family, named $\chan^{\kappa}_{\SES}(\ell)$, 
indexed with parameter $\kappa\in \nats$, where the documents  will correspond
to the input of generalized algorithm substitution attack against encryption schemes.
The essential idea behind the definition of the channel $\chan^{\kappa}_{\SES}(\ell)$ 
is that for all $k\in \supp(\Gen(1^{\kappa}))$ and every sequence of messages
$m_{1},m_{2},\ldots, m_{\ell(\kappa)}$, with $m_{i}\in\{0,1\}^{\ml(\kappa)}$, for 
the history 
\[
   h= k \mid\mid m_{1} \mid\mid m_{2} \mid\mid \ldots \mid\mid m_{\ell(\kappa)}
\]
the distribution of the sequences of documents 
  \begin{align*}
     c_{1} \mid\mid c_{2} \mid\mid \ldots \mid\mid c_{\ell(\kappa)}
  \end{align*}
generated by the channel is exactly the same as the distribution for
  \begin{align*}
     \Enc(k,m_{1}) \mid\mid \Enc(k,m_{2}) \mid\mid \ldots \mid\mid \Enc(k,m_{\ell(\kappa)}).
  \end{align*}
To give a formal definition of $\{\chan^{\kappa}_{\SES}(\ell)\}_{\kappa\in \nats}$
we need to specify the probability distributions for any history $h$. 
Thus, we define the family, on the alphabet $\{0,1\}$, 
as follows.

For the empty history $h=\varnothing$, define
\begin{align*}
\chan^{\kappa}_{\SES}(\ell)_{\varnothing}  
\end{align*}
as the 
distribution of all keys generated by $\Gen(1^{\kappa})$.
For a key $k\in \supp(\Gen(1^{\kappa}))$
  and a (possibly empty) sequence of messages
  $m_{1},m_{2},\ldots, m_{r}$, with $m_{i}\in \{0,1\}^{\ml(\kappa)}$ and
  $0\leq r\leq \ell(\kappa)-1$,  the distribution
  \begin{align*}
    \chan^{\kappa}_{\SES}(\ell)_{k\mid\mid m_{1}\mid\mid
    m_{2}\mid\mid \ldots \mid\mid m_{r}}
  \end{align*}
  is the uniform distribution on all messages $m_{r+1}\in
  \{0,1\}^{\ml(\kappa)}$.
For $k\in \supp(\Gen(1^{\kappa}))$, a sequence of messages
  $m_{1},m_{2},\ldots,$ $ m_{\ell(\kappa)}$ with $m_{i}\in
  \{0,1\}^{\ml(\kappa)}$, and a  (possibly empty) sequence of
  ciphertexts $c_{1},\ldots,c_{r}$,  with 
  $c_{i}\in \supp(\Enc(k,m_{((i-1)\bmod
    \ell(\kappa))+1})),$ the distribution
  \begin{align*}
    \chan^{\kappa}_{\SES}(\ell)_{k\mid\mid m_{1}\mid\mid
    m_{2}\mid\mid \ldots \mid\mid m_{\ell(\kappa)}\mid\mid
    c_{1}\mid\mid c_{2}\mid\mid \ldots \mid\mid
    \ldots  \mid\mid c_{r}}
  \end{align*}
  is the distribution of $\Enc(k,m_{(r \bmod \ell(\kappa))+1})$.

  \section{ASAs against Encryption in the Steganographic Model}
\label{Sec:ASA:against:encrypted:as:stego}
The main message of our paper is that algorithm
substitution attacks against a primitive $\Pi$ are equivalent to
the use of steganography on a corresponding channel $\chan_{\Pi}$
determined by the protocol $\Pi$.
Focusing on symmetric encryption schemes as a common cryptographic primitive,
we will show in this section exemplary proofs for the general relations
between ASAs and steganography.

In the previous section we showed a formal specification of the family of communication
channels $\chan^{\kappa}_{\SES}(\ell)$ determined by a symmetric encryption scheme  
$\SES$. We will now prove that a secure and reliable stegosystem 
on $\chan^{\kappa}_{\SES}(\ell)$ implies the existence of an indistinguishable and
successful algorithm substitution attack on $\SES$. 
On the other hand, we will also show that the existence of an
indistinguishable and successful algorithm substitution attack on $\SES$
implies a secure and reliable stegosystem on $\chan^{\kappa}_{\SES}(\ell)$. 

As a consequence we get a construction of an ASA against any encryption 
scheme using a generic stegosystem like \eg this proposed by 
Dedi{\'c} et al.~\cite{dedic2009upper}. Thus, we can conclude  
Theorem~1 and Theorem~3 proposed by Bellare et al.~in \cite{bellare2014asa}
that there exist indistinguishable and successful ASAs against encryption schemes.
Moreover we obtain Theorem~4 in \cite{bellare2014asa} which says
that an ASA is impossible for unique ciphertext symmetric encryption schemes.

\subsection{Steganography implies ASAs} 
\begin{theorem}\label{thm:asa:on:ses:impl:stego}
Assume $\SES$ is a symmetric encryption scheme and let $\Steg$ be a stegosystem 
on the channel $\chan:=\chan^{\kappa}_{\SES}(\combine{\Steg}{\outl}(\kappa))$ 
determined by $\SES$. 
Then there exists an algorithm substitution
attack $\ASAE$ against $\SES$ of indistinguishability, resp. reliability such that:
\[
\begin{array}{rcl}
   \insec^{\encwatch}_{\ASAE,\SES}(\kappa) &\le&  \insec^{\cha}_{\Steg,\chan}(\kappa)
   \quad \text{and} \\
   \unrel_{\ASAE,\SES}(\kappa) &=&  \unrel^{\star}_{\Steg,\chan}(\kappa).
\end{array}
\]
\end{theorem}

\begin{proof}
  Let $\SES=(\Gen,\Enc,\Dec)$ be a symmetric encryption scheme
  and $\Steg=(\SGen,\SEnc,\SDec)$ be a stegosystem on 
  the channel $\chan$. To simplify notation, let
  $\ell=\ell(\kappa):=\combine{\Steg}{\outl}(\kappa)$.
  We will construct the algorithm substitution attack $\ASAE=(\AGen,\AEnc,$ $\AExt)$ on
  $\SES$ from the stegosystem $\Steg$ 
  and show the indistinguishability and
  success of $\ASAE$ depending on 
  security and reliability of $\Steg$. 
  The components of the $\ASAE$ are defined as follows.
 
    The key generator $\AGen$ just simulates $\SGen$ -- the key
    generator of the stegosystem. It will output the attack key $\ak$.
    The encoding algorithm $\AEnc$ on input $\ak\in
    \supp(\AGen(1^{\kappa}))$, $\am\in
    \{0,1\}^{\combine{\Steg}{\ml}(\kappa)}$, $k\in
    \supp(\Gen(1^{\kappa}))$,  and
    $m\in
    \{0,1\}^{\combine{\SES}{\ml}(\kappa)}$ simulates $\SEnc$ 
    on channel $\chan$ 
    with input
    key $\ak$, the message $\am$ and the history $h=k\mid\mid m^{\ell}$,
    where $m^{\ell}$ is the   string of length $\ell\cdot |m|$
    containing $\ell$ copies of $m$.
   Whenever $\SEnc$ makes a query
    to its channel oracle, algorithm $\AEnc$ 
    uses $\Enc$ on input $k$ and $m$ to produce a corresponding ciphertext and
    sends it to $\SEnc$. The encoder $\AEnc$ then outputs the 
    document  $d$ generated by $\SEnc$. 
    Finally, the extraction algorithm $\AExt$ on input
    $\ak\in \supp(\AGen(1^{\kappa}))$ and documents 
    $d_{1},\ldots,d_{\ell}$ just simulates
    $\SDec$ on the same inputs. 

  As one can see from the definitions, $\ASAE$ is a generalized
  algorithm substitution attack against $\SES$. We will now prove that
  it is indistinguishable from $\SES$ and that it is successful. 

  We prove first indistinguishability of the system. 
  Let $\Watch$ be a watchdog against the above 
    $\ASAE$ with maximal advantage, i.\,e. 
    \begin{align*}
     \adv_{\Watch,\ASAE,\SES}^{\encwatch}(\kappa) = \insec^{\encwatch}_{\ASAE,\SES}(\kappa),
    \end{align*}
    where 
  $\adv^{\encwatch}_{\Watch, \ASAE,\SES}(\kappa)$ is equal to  the success
probability that 
$\EncASA-Dist_{\Watch,\ASAE,\SES}(\kappa)=\textsf{true}$.
We will now construct a warden $\Ward$ from $\Watch$ such that
\begin{align*}
\adv^{\cha}_{\Ward,\Steg,\chan}(\kappa)
=
\adv_{\Watch,\ASAE,\SES}^{\encwatch}(\kappa).  
\end{align*}
Thus, we will get that 
\begin{equation}\label{ineq:insecasa:insecstego}
  \insec^{\encwatch}_{\ASAE,\SES}(\kappa)
  \le 
  \insec^{\cha}_{\Steg,\chan}(\kappa).
\end{equation}
The warden
$\Ward$ on input $1^{\kappa}$ just simulates the watchdog
$\Watch$ and gives the same output as $\Watch$ at the end of the simulation.
Whenever the watchdog makes a query on input $\am$,
$k$,  and $m$
to its challenging oracle (that is either equal to
$\SES$'s encryption algorithm 
$\Enc(k,m)$ or to 
$\ASA$'s encryption  
$\AEnc(\ak,\am,k,m,\sigma)$
for $\ak\gets \AGen(1^{\kappa})$), the warden $\Ward$ queries its own
challenging oracle with message $\am$, state $\sigma$ and history 
$h=k\mid\mid m^{\ell}$.
Note that
the challenging oracle of $\Ward$ is either equal to the channel  $\chan$ or to
$\SEnc(\ak,\am,h,\sigma)$ for $\ak\gets \SGen(1^{\kappa})$.

If the challenging oracle of $\Ward$ is equal to the steganographic encoding 
$\SEnc(\ak,\am,h,\sigma)$ (\ie the
bit $b$ in  $\SS-CHA-Dist$ equals $1$, denoted by
$\SS-CHA-Dist_{\Ward,\Steg,\chan}(\kappa)\langle b=1\rangle$), the
answer of $\Ward$
is the same as the output of the $\Watch$ 
in case it queries the ASA's encoding algorithm
$\AEnc(\ak,\am,k,m)$ 
by construction. Thus, 
\begin{align*}
  &\Pr[\SS-CHA-Dist_{\Ward,\Steg,\chan}(\kappa)\langle b=1\rangle=\textsf{true}]\\  
  &= \ \Pr[\EncASA-Dist_{\Watch,\ASAE,\SES}(\kappa)\langle b=1\rangle=\textsf{true}].
\end{align*}

If the challenging oracle of $\Ward$ is equal to the channel (the
bit $b$ in $\SS-CHA-Dist$ equals $0$), by the definition of 
the channel $\chan$ for the symmetric encryption scheme $\SES$, the
answer of the challenging oracle is equal to the output of
$\Enc(k,m)$.  Hence,
\begin{align*}
  &\Pr[\SS-CHA-Dist_{\Ward,\Steg,\chan}(\kappa)\langle b=0\rangle=\textsf{true}]\\
  & = \
  \Pr[\EncASA-Dist_{\Watch,\ASAE,\SES}(\kappa)\langle b=0\rangle=\textsf{true}].
\end{align*}
We thus have
\[
\begin{array}{rcl}
  \adv^{\cha}_{\Ward,\Steg,\chan}(\kappa)
    & = &
  |\Pr[\SS-CHA-Dist_{\Ward,\Steg,\chan}(\kappa)=\textsf{true}] - 1/2| \\
  & = & |\Pr[\EncASA-Dist_{\Watch,\ASAE,\SES}(\kappa)=\textsf{true} -1/2| \\ 
  & = &
  \adv^{\encwatch}_{\Watch, \ASAE,\SES}(\kappa)
\end{array}
\]
which completes the proof of \eqref{ineq:insecasa:insecstego}.

We still need to prove that $\AExt$ is reliably able to
  extract the attacker message $\am$ from the ciphertext. But, as
  $\AExt=\SDec$, the reboot-reliability of $\SDec$ directly implies that
  $\AExt$ is successful with probability of $1-\negl(\kappa)$.
\end{proof}

By combining Theorem~\ref{thm:asa:on:ses:impl:stego} and
Theorem~\ref{thm:rejsam:secure}, we can conclude the following
corollary. 
\begin{corollary}\label{coroll:thm:ipl:bellare}
  For every symmetric encryption scheme $\SES$, there exists an algorithm
  subsection attack $\ASA$ with message length $\ml(\kappa)$ and
  parameter $s\ge 1$ such that 
  \[ 
   \begin{array}{rcl}
     \insec^{\encwatch}_{\ASAE,\SES}(\kappa)&\leq&
       \mathcal{O}(\ml(\kappa)^{4}\cdot 2^{-\minent(\chan^{\kappa})})+\\
&&\mathcal{O}(\ml(\kappa)^{2}\cdot
    \exp(-s))+
     \InSec^{\prf}_{\algf}(\kappa),\\[2mm]
\unrel^{\star}_{\Steg,\chan}(\kappa) &\leq&
    2 \ml(\kappa)^{2}\cdot
    \exp(-2^{\minent(\chan^{\kappa})-3}) +\\
&&  \ml(\kappa)^{2}\cdot  \exp(-2^{-2} s)+
    \InSec^{\prf}_{\algf}(\kappa)
    \end{array}
  \]
where $\chan:=\chan^{\kappa}_{\SES}(\combine{\Steg}{\outl}(\kappa))$ 
\end{corollary}

One can compare this corollary to the construction used in the proof of
Theorem~4.1 and Theorem~4.2 in \cite{bellare2015asa}. We can see that our generic 
algorithm substitution attack gets almost the same bounds for insecurity
and for unreliability.

Note that the protocols in \cite{bellare2015asa, ateniese2015sig}
and our generic protocol of Corollary \ref{coroll:thm:ipl:bellare}
have a very bad rate: 
$
    \frac{\ml}{\ml\cdot (\ln \ml +\beta)} = 1/(\ln \ml +\beta)
$
for an appropriate value $\beta$. One can easily modify the above  
constructions such that instead of one bit $b$ of a message $\am$ we embed
a block of $\log(\ml)$ bits per ciphertext.  
This improves the rate to 
$
     \frac{\log \ml}{\ln(\ml) - \ln\log(\ml) +\beta} = \Theta(1).
$

\subsection{ASAs imply Steganography}   

\begin{theorem}\label{thm:ses:on:ses:impl:stego}
Assume $\SES$ is a symmetric encryption scheme and let 
$\ASAE$ be an algorithm substitution attack against $\SES$ 
of output length $\combine{\ASAE}{\outl}(\kappa)$. 
Then there exists a stegosystem $\Steg$  with the output length
$\combine{\Steg}{\outl}(\kappa)=2\cdot  \combine{\ASAE}{\outl}(\kappa)+1$
on the channel $\chan:=\chan^{\kappa}_{\SES}(\combine{\Steg}{\outl}(\kappa))$ 
determined by $\SES$ 
such that $\Steg$'s insecurity, resp. its reliability satisfy
\[
\begin{array}{rcl}
   \insec^{\cha}_{\Steg,\chan}(\kappa) &\le& \insec^{\encwatch}_{\ASAE,\SES}(\kappa)
   \quad \text{and} \\
   \unrel_{\Steg,\chan}(\kappa) &=& \unrel_{\ASAE,\SES}(\kappa) . 
\end{array}
\]
\end{theorem}

\begin{proof}
  Let $\SES=(\Gen,\Enc,\Dec)$ be a symmetric encryption scheme
  and  $\ASAE=(\AGen,\AEnc,\AExt)$ be an algorithm
  substitution attack against $\SES$.
  To simplify notation, let  $\ell=\combine{\ASAE}{\outl}(\kappa)$. 
  We construct the stegosystem $\Steg=(\SGen,\SEnc,\SDec)$ on $\chan$ out of the  $\ASAE$.

   The key generation algorithm $\SGen$ simply simulates
    $\AGen$. It will output the key $\ak$.
    To encode a message $\am$ using the key $\ak$,
    the stegoencoding algorithm $\SEnc$  generates for any history $h$
    a sequence of $\combine{\Steg}{\outl}(\kappa)=2\ell+1$ documents
    such that the last $\ell$ documents embed the message $\am$.
    To describe the algorithm we 
    need to distinguish between
    different given histories~$h$.
    \begin{compactdesc} 
    \item[$h=\varnothing$:] In this case,  $\SEnc$ chooses a random key
      $k\gets \combine{\SES}{\Gen}(1^{\kappa})$ using the generation algorithm of $\SES$
      and outputs $k$.

\item[\rm $h=k\mid\mid m_1\mid\mid m_2\mid\mid \ldots \mid\mid
  m_r$ for $0\leq r\leq \ell-1$:]  Encoder $\SEnc$  samples 
  a random message $m_{r+1}$ and outputs it. 

\item[\rm $h=k\mid\mid m_{1}\mid\mid m{_2}\mid\mid \ldots
  \mid\mid m_{\ell}\mid\mid c_{1}\mid\mid \ldots \mid\mid
  c_{r}$ with $r\geq 0$:] The stego-encoder $\SEnc$ simulates
  $\AEnc(\ak,\am,k,m_{(r+1) \bmod \ell +1})$ and outputs the generated ciphertext.
\end{compactdesc}

Note that by construction, in any case the last $\ell$ documents generated by $\SEnc^{2\ell+1}$ 
embed the message $\am$ in the same way as done by $\ASAE^{\ell}$. 

If the decoder $\SDec$ is given documents $d_{1},\ldots,d_{2\ell+1}$, we
  output $\AExt(\ak,d_{\ell+2},\ldots,d_{2\ell+1})$. 

  As one can see from the definitions, the decoding algorithm of $\Steg$ is  history-ignorant.
  We will prove that on the channel  $\chan=\chan^{\kappa}_{\SES}(2\ell+1)$ 
  the security and reliability of the stegosystem $\Steg$ 
  satisfy the stated conditions.
  
We first analyze the security of the system.
 Let $\Ward$ be a warden against $\Steg$ on $\chan$ with maximal
    advantage, i.\,e.
    \begin{align*}
     \adv^{\cha}_{\Ward,\Steg,\chan}(\kappa)
      =
      \insec^{\cha}_{\Steg,\chan}(\kappa),
    \end{align*}
   where
   $ 
      \adv^{\cha}_{\Ward,\Steg,\chan}(\kappa)
        = \Pr[\SS-CHA-Dist_{\Ward,\Steg,\chan}(\kappa)=\textsf{true}].
    $ 
    We will  construct a watchdog $\Watch$ against  the algorithm
    substitution attack $\ASAE$ with the same advantage as $\Ward$:
    \begin{align*}
      \adv_{\Watch,\ASAE,\SES}^{\encwatch}(\kappa) =
      \adv^{\cha}_{\Ward,\Steg,\chan}(\kappa).
    \end{align*}
   This will prove that 
   \begin{equation}\label{ineq:insecstego:insecasa}
    \insec^{\cha}_{\Steg,\chan}(\kappa)
    \le
      \insec^{\encwatch}_{\ASAE,\SES}(\kappa).
    \end{equation}

    The watchdog $\Watch$ on input $1^{\kappa}$ simply
    simulates the warden $\Ward$. Whenever the warden $\Ward$ makes a
    query to its channel oracle $\chan$ with a history $h$, the watchdog 
    $\Watch$ simulates the oracle response as follows:
    \begin{itemize}
    \item If $h=\varnothing$, the watchdog uses
      $\Gen(1^{\kappa})$ to construct a key $k$ and
      returns $k$ to the  warden.
    \item If $h=k\mid\mid m_{1}\mid\mid \ldots \mid\mid
      m_{r}$ with $r < \ell$, the watchdog uniformly chooses a
      message $m_{r+1}$ from $\{0,1\}^{\combine{\SES}{\ml}(\kappa)}$ and
      outputs $m_{r+1}$.
   \item If $h=k\mid\mid m_{1}\mid\mid \ldots\mid\mid 
      m_{\ell}\mid\mid c_{1}\mid\mid \ldots\mid\mid c_{r}$ with
      $r\geq 0$, the watchdog computes
      $c_{r+1}\gets \Enc(k,m_{((r+1)\bmod \ell)+1})$ and
      outputs $c_{r+1}$.
    \end{itemize}
    Clearly, this simulates the channel distribution
    $\chan$ 
    perfectly. If the warden queries its challenge oracle
    $\combine{\Ward}{\CH}$ with chosen message $\am$, state $\sigma$, and history $h$ (that is
    either equivalent to
    sampling from $\chan_{h}$
    or to calling $\SEnc(\ak,\am,h,\sigma)$), the watchdog simulates the response 
    of the oracle $\combine{\Ward}{\CH}$ as follows:
 \begin{itemize}
  \item If $h=\varnothing$ then $\Watch$ chooses a random key
      $k\gets \Gen(1^{\kappa})$ and outputs it.
  \item If $h=k\mid\mid m_1\mid\mid m_2\mid\mid \ldots \mid\mid
      m_r$ for $0\leq r\leq \ell-1$ then $\Watch$ samples 
      a random message
      $m$ and outputs it.
  \item If $h=k\mid\mid m_{1}\mid\mid m{_2}\mid\mid \ldots
      \mid\mid m_{\ell}\mid\mid c_{1}\mid\mid \ldots \mid\mid
      c_{r}$ with $r\geq 0$ then 
      $\Watch$ queries its own oracle on $k$ and $m_{((r+1)\bmod \ell)+1}$.
\end{itemize}
    If $\combine{\Watch}{\CH}$ is equal to
    $\Enc$ of $\SES$ (the bit $b$ in $\EncASA-Dist$ is
    set to $0$) the corresponding answer is identically
    distributed to a  sample of the channel $\chan$. Hence,
    \begin{align*}
      &\Pr[\EncASA-Dist_{\Watch,\ASAE,\SES}(\kappa)\langle
      b=0\rangle=\textsf{true}]=\\
      &\Pr[\SS-CHA-Dist_{\Ward,\Steg,\chan}(\kappa)\langle b=0\rangle=\textsf{true}].
    \end{align*}
    On the other hand, if $\combine{\Watch}{\CH}$ is equal to
    $\AEnc$ (the bit $b$ in $\EncASA-Dist$ is set
    to $1$), the corresponding answer is identically
    distributed to $\SEnc(\ak,\am,h,\sigma)$ and thus
    \begin{align*}
       & \Pr[\EncASA-Dist_{\Watch,\ASAE,\SES}(\kappa)\langle b=1\rangle=\textsf{true}]=\\
      &\Pr[\SS-CHA-Dist_{\Ward,\Steg,\chan}(\kappa)\langle b=1\rangle=\textsf{true}].
    \end{align*}
    We thus have
    \[
       \begin{array}{ll}
         & \adv^{\encwatch}_{\Watch,\ASAE,\SES}(\kappa)
          \ =\\
         &\quad\quad |\Pr[\EncASA-Dist_{\Watch,\ASAE,\SES}(\kappa)=\textsf{true}] - 1/2|  \ =\\
         &\quad\quad |\Pr[\SS-CHA-Dist_{\Ward,\Steg,\chan}(\kappa)=\textsf{true}] - 1/2|  \ =\\
         &\quad\quad \adv^{\cha}_{\Ward,\Steg,\chan}(\kappa)
      \end{array}
    \]
    which proves \eqref{ineq:insecstego:insecasa}.

           The reliability of  $\Steg$ 
    is the same as the success probability of $\ASAE$ since 
     $\SDec$ simply simulates
    $\AExt$. 
\end{proof}

By using the fact that channels with min-entropy $0$ can not be used for
steganography (see \eg Theorem~6 in \cite{hopper2009provably}) and
observing that channels corresponding to deterministic encryption
schemes have min-entropy $0$, we can conclude the following corollary:

\begin{corollary}
For all deterministic encryption schemes $\SES$ and all
algorithm substitution attacks $\ASA$ against $\SES$:
\begin{align*}
    \insec^{\encwatch}_{\ASAE,\SES}(\kappa)\geq 1.
\end{align*}
\end{corollary}
Note that this exactly Theorem~4 in \cite{bellare2014asa}.

\section{General Results}
\label{sec:general}

Let $\R$ be a polynomial-time randomized algorithm with hardwired secret
$s$ which takes inputs $x$ and generates outputs $y$.  The general task
of an algorithm substitution attack against $\R$ is to construct a
subverted algorithm $\AR_{\ak}$ which using a hidden hardwired random
key $\ak$ outputs on the secret $s$ in the sequence of calls
$\AR_\ak(s,x_1), \AR_\ak(s,x_2),\ldots$ a sequence such that
\begin{compactenum} 
\item the output $\AR_\ak(s,x_1), \AR_\ak(s,x_2),\ldots$ is indistinguishable from
      $\R(s,x_1), \R(s,x_2),\ldots$  and 
\item $\AR_\ak(s,x_1), \AR_\ak(s,x_2),\ldots$  embeds the secret $s$.
\end{compactenum}

In our setting we model the attack on $\R$ as a stegosystem 
on a channel determined by $\R$ and define such a channel. 

\subsection{ASA against a Randomized Algorithm}

In this section  we give formal definitions for 
algorithm substitution attack $\AR$, 
its advantage $\adv_{\Watch,\AR,\R}$, etc\@. Formally, an  \emph{algorithm substitution attack  against $\R$}
is a triple of efficient algorithms $\ASA=(\Gen,\AR,\Ext)$, where $\Gen$ generates the
key $\ak$, the algorithm $\AR$ takes the key $\ak$, a secret $s$ and all
inputs $x_{1},x_{2},\ldots$ to $\R$ and the extractor $\Ext$ tries to extract $s$ from the
outputs of $\AR$ with the help of $\ak$ (but without knowing $x_{1},x_{2},\ldots$). Similarly to the setting for encryption
schemes, $\ASA$ is called \emph{indistinguishable}, if every PPTM
$\Watch$ -- the \emph{watchdog} -- is not able to distinguish between
$\AR_{\ak}(s,x_{1}),\AR_{\ak}(s,x_{2}),\ldots$ and 
$\R(x_{1}),\R(x_{2}),\ldots$ even if he is allowed to choose $s$ and all
$x_{i}$. This is defined via the game
$\RASA-Dist_{\Watch,\ASA,\R}$ defined analogously to
$\EncASA-Dist$. The maximal advantage of any watchdog distinguishing
$\ASA$ from $\R$ is
called the \emph{insecurity} or indistinguishability of $\ASA$ and is
formally defined as
\begin{align*}
\insec^{\operatorname{asa}}_{\ASA,\R}(\kappa) := \max_{\Watch}\{\adv_{\Watch,\ASA,\R}^{\operatorname{asa}}(\kappa)\},
\end{align*}
where
\begin{align*}
    &\adv^{\operatorname{asa}}_{\Watch, \ASA,\R}(\kappa) 
   := \\ 
   &\quad\quad |\Pr[\RASA-Dist_{\Watch,\ASA,\R}(\kappa)=\textsf{true}] - 1/2|.
\end{align*}

The \emph{unreliability} of $\ASA$ is also defined like before:
\begin{align*}
  &\unrel_{\ASA,\R}(\kappa) \ :=\\ &\max
  \{
  \Pr[\combine{\ASAE}{\AExt}(\ak,\combine{\ASAE}{\AR}(\ak,\am,x_{1},\ldots,x_{\ell}))\neq
  \am]\},
\end{align*}
where the maximum is taken over all
$\ak\in \supp(\combine{\ASAE}{\Gen}(1^{\kappa})), \allowbreak
 \am \in\{0,1\}^{\combine{\ASAE}{\ml}(\kappa)}$,
and   $x_{1},\ldots,x_{\ell}$ being inputs to $\R$.

Known examples which fit into this setting include \eg the
subversion-resilient signature schemes presented 
in the work of Ateniese et al.~\cite{ateniese2015sig}.

\subsection{Channel determined by a Randomized Algorithm}
 
Let $\R$ be a polynomial-time randomized algorithm  with parameter~$\kappa$. 
We assume that the secret $s$ is generated by $\Gen$ and the inputs $x$ to $\R$  
are generated by the randomized polynomial-time algorithm $\Geni$,
associated with $\R$ (which may be chosen adversarially as shown in the
definition above). 
Let $\ell$ be a polynomial of $\kappa$.
For $\R$ we define a channel family, named $\chan^{\kappa}_{\R}(\ell)$, 
indexed with parameter $\kappa\in \nats$, with documents which correspond
to the input of $\AR$.
The essential idea behind the definition of the channel $\chan^{\kappa}_{\R}(\ell)$ 
is that for all $s\in \supp(\Gen(1^{\kappa}))$ and every sequence of inputs
$x_{1},x_{2},\ldots, x_{\ell(\kappa)}$, with $x_{i}\in\supp(\Geni(1^{\kappa}))$, for 
the history 
\[
   h= s \mid\mid x_{1} \mid\mid x_{2} \mid\mid \ldots \mid\mid x_{\ell(\kappa)}
\]
the distribution of the sequences of documents 
  \begin{align*}
     y_{1} \mid\mid y_{2} \mid\mid \ldots \mid\mid y_{\ell(\kappa)}
  \end{align*}
generated by the channel is exactly the same as the distribution for
  \begin{align*}
     \R(s,x_{1}) \mid\mid \R(s,x_{2}) \mid\mid \ldots \mid\mid \R(s,x_{\ell(\kappa)}).
  \end{align*}
To give a formal definition of $\{\chan^{\kappa}_{\R}(\ell)\}_{\kappa\in \nats}$
we need to specify the probability distributions for any history $h$. 
Thus, we define the family, on the alphabet $\{0,1\}$, 
as follows: For empty history $h=\varnothing$, we define $\chan^{\kappa}_{\R}(\ell)_{\varnothing}$ as the 
  distribution on all possible keys generated by $\Gen(1^{\kappa})$.
For $s\in \supp(\Gen(1^{\kappa}))$
  and a (possibly empty) sequence inputs 
  $x_{1},x_{2},\ldots, x_{r}$ with $x_{i}\in \supp(\Geni(1^{\kappa}))$ and
  $0\leq r\leq \ell(\kappa)-1$,  the distribution
      $\chan^{\kappa}_{\R}(\ell)_{s\mid\mid x_{1}\mid\mid
    x_{2}\mid\mid \ldots \mid\mid x_{r}}$
  is the distribution on inputs  $x_{r+1}\gets \Geni(1^{\kappa})$.
 For $s\in \supp(\Gen(1^{\kappa}))$, a sequence of inputs
  $x_{1},x_{2},\ldots, x_{\ell(\kappa)}$ with $x_{i}\in\supp(\Geni(1^{\kappa}))$,
  and a  (possibly empty) sequence of
  $\R$'s outputs $y_{1},\ldots,y_{r}$ with $y_{i}\in \supp(\R(s,x_{((i-1)\bmod
    \ell(\kappa))+1}))$, the probability distribution of
  $\chan^{\kappa}_{\R}(\ell)_{s\mid\mid x_{1}\mid\mid
    x_{2}\mid\mid \ldots \mid\mid x_{\ell(\kappa)}\mid\mid
    y_{1}\mid\mid y_{2}\mid\mid \ldots \mid\mid
    \ldots  \mid\mid y_{r}}$
  is the probability distribution of $\R(s,x_{(r \bmod \ell(\kappa))+1})$.

\subsection{Results}

The theorems proved in the previous section can simply be generalized
by using our general construction of the channel $\chan^{k}_{\R}(\ell)$
for the randomized algorithm $\R$ and the generic stegosystem
$\RejSam^{\algf}$ 
provided by Theorem~\ref{thm:rejsam:secure}. 
\begin{theorem}
  \label{thm:generic-attack:against:R}
  For every randomized algorithm $\R$, there exists a generic algorithm
  substitution attack $\ASA$ against $\R$ such that
  \[
    \begin{array}{rcl}
      \insec^{\operatorname{\ASA}}_{\ASAE,\R}(\kappa)&\leq&
 \mathcal{O}(\ml(\kappa)^{4}\cdot 2^{-\minent(\chan^{\kappa})})+\\
&&\mathcal{O}(\ml(\kappa)^{2}\cdot
    \exp(-s))+
     \InSec^{\prf}_{\algf}(\kappa),\\[2mm]
      \unrel^{\star}_{\Steg,\chan}(\kappa) &\leq&
    2 \ml(\kappa)^{2}\cdot
    \exp(-2^{\minent(\chan^{\kappa})-3}) +\\
&&  \ml(\kappa)^{2}\cdot  \exp(-2^{-2} s)+
    \InSec^{\prf}_{\algf}(\kappa)                
  \end{array}
  \] 
  where $\chan:=\chan^{\kappa}_{\R}(\combine{\Steg}{\outl}(\kappa))$.
\end{theorem}

\begin{theorem}
  \label{thm:no-attack:against:R}
  For all deterministic algorithms $\R$ and all
algorithm substitution attacks $\ASA$ against~$\R$:
\begin{align*}
    \insec^{\operatorname{asa}}_{\ASAE,\R}(\kappa) = 1.
\end{align*}

\end{theorem}
Theorem~\ref{thm:generic-attack} is thus just a consequence of
Theorem~\ref{thm:generic-attack:against:R} and
Theorem~\ref{thm:min-entropy} is just a consequence of
Theorem~\ref{thm:no-attack:against:R}.

These general results also imply several other results from the
literature, for example on signature schemes.
Ateniese et al.~\cite{ateniese2015sig} study algorithm substitution
  attacks\footnote{To be more precise, their attacks
    only replace the signing algorithm $\Sign$.} on \emph{signature schemes}
  $\SIG=(\Gen,\Sign,\Vrfy)$, where
  \begin{itemize}
  \item The \emph{key generator} $\combine{\SIG}{\Gen}$ produces upon
    input $1^{\kappa}$ a pair $(\pk,\sk)$ of keys with
    $|\pk|=|\sk|=\kappa$. We call $\pk$ the \emph{public key} and $\sk$ the
    \emph{secret key}.
\item The \emph{signing algorithm} $\combine{\SIG}{\Sign}$ takes as
  input the secret key $\sk$ and a message
  $m\in \{0,1\}^{\combine{\SIG}{\ml}}$ of length
  $\combine{\SIG}{\ml}(\kappa)$ and produces a signature
  $\sigma\in \{0,1\}^{\combine{\SIG}{\sigl}(\kappa)}$ of length
  $\combine{\SIG}{\sigl}(\kappa)$.
\item The \emph{verifying algorithm} $\combine{\SIG}{\Vrfy}$ takes as
  input the public key $\pk$, the message $m$ and a signature $\sigma$
  and outputs a bit~$b$.
\end{itemize}

On the positive side (from the view of an algorithm substitution attack)
they show that all randomized \emph{coin-injective} schemes and all
\emph{coin-extractable} schemes have ASA\@. A randomized algorithm $A$ is
\emph{coin-injective}, if the function $f_{A}(x,\rho)=A(x;\rho)$ (where
$\rho$ denotes the random coins used by $A$) is injective and
\emph{coin-extractable} if there is another randomized algorithm $B$
such that $\Pr[B(A(x;\rho))=\rho] \geq 1-\negl$ for a negligible
function $\negl$. They prove the following theorems:
  \begin{theorem}[Theorem 1 in \cite{ateniese2015sig}]
    For every coin-injective signature scheme $\SIG$, there is a successful
    algorithm substitution attack $\ASA$ and a negligible function
    $\negl$ such that
    \begin{align*}
      \insec^{\operatorname{asa}}_{\ASA,\SIG}(\kappa)\leq \insec^{\operatorname{prf}}_{\mathsf{F}}(\kappa)+\negl(\kappa)
    \end{align*}
    for a pseudorandom function $\mathsf{F}$. 
  \end{theorem}
  \begin{theorem}[Theorem 2 in \cite{ateniese2015sig}]
    For every coin-extractable signature scheme $\SIG$, there is a successful
    algorithm substitution attack $\ASA$ and a negligible function
    $\negl$ such that
    \begin{align*}
            \insec^{\operatorname{asa}}_{\ASA,\SIG}(\kappa)\leq \negl(\kappa).
    \end{align*}
  \end{theorem}
Both of these results are easily implied by
Theorem~\ref{thm:generic-attack:against:R}.

On the negative side (from the view of an algorithm substitution
attack), they show that \emph{unique signature schemes} are resistant to
ASAs fulfilling the \emph{verifiability condition}. Informally this
means that (a) each message has exactly on signature (for a fixed
key-pair) and (b) each signature produced by the ASA must be valid.
\begin{theorem}[Theorem 3 in \cite{ateniese2015sig}]
  For all unique signature schemes $\SIG$ and all algorithm substitution
  attacks $\ASA$ against them that fulfill the verifiability condition,
  there is a negligible function $\negl$ such that
      \begin{align*}
            \insec^{\operatorname{asa}}_{\ASA,\SIG}(\kappa)\geq 1-\negl(\kappa).
    \end{align*}
  \end{theorem}
  As unique signature schemes do not provide enough min-entropy for a
  stegosystem, this results follows from Theorem~\ref{thm:min-entropy}.

\section{A Lower Bound for Universal ASA}
\label{sec:bound}

A setting similar to steganography, where \emph{universal} stegosystems
exist, that can be used for \emph{any} channel of sufficiently large
min-entropy, would be quite useful for attackers that plan to launch
algorithm substitution attacks. Such a system would allow them to attack
any symmetric encryption scheme \emph{without} knowing the internal specification of the
encryption algorithm. 
A closer look at the results in
\cite{bellare2014asa, bellare2015asa, ateniese2015sig} reveals that
their attacks do indeed go without internal knowledge of the used
encryption algorithm. They only manipulate the random coins used in the
encryption process. Note that $\combine{\SES}{\Enc}(k,m;r)$ (where $r$
denotes the random coins used by $\Enc$) is a deterministic function, as
$\combine{\SES}{\Enc}$ is a PPTM. 

We thus define a  \emph{universal algorithm substitution attack} as a triple of PPTMs
such that for every symmetric encryption scheme
$\SES$, the triple
\begin{align*}
\ASA^{\SES}=(\combine{\ASA}{\Gen},\combine{\ASA}{\Enc}^{\combine{\SES}{\Enc}(\cdot,\cdot; \cdot)},\combine{\ASA}{\Ext})  
\end{align*}
is an ASA against $\SES$. Hence,
$\combine{\ASA}{\Enc}$ has only oracle access to the encryption
algorithm $\combine{\SES}{\Enc}$ of the encryption scheme: It may thus choose
arbitrary values $k$, $m$, and $r$ and receives a 
ciphertext
\begin{align*}
c\gets \combine{\SES}{\Enc}(k,m; r)  
\end{align*}
without having 
a complete description of the encryption schemes.

As noted above, all attacks in
\cite{bellare2014asa, bellare2015asa, ateniese2015sig} are 
universal and \citeauthor{bellare2015asa} explicitly state in their
work~\cite{bellare2015asa} that their ASA works against any encryption
scheme of sufficiently large min-entropy. We also remark that the
rejection sampling ASA presented earlier is universal.

For a universal algorithm substitution attack $\ASA$ and a
symmetric encryption scheme $\SES$, let
$\combine{\ASA}{\query}(\SES,\kappa,\ak,\am,k,m_j,\sigma)$ be the
expected number of oracle calls that a single call of the substitution encoder 
$\combine{\ASA}{\Enc}^{\combine{\SES}{\Enc}(\cdot,\cdot;\cdot)}(\ak,\am,k,m_j,\sigma)$
makes to its encryption oracle $\combine{\SES}{\Enc}$. We then define
\begin{align*}
  &\combine{\ASA}{\query}(\SES,\kappa)=\\
  &\max_{\substack{\ak\in\supp(\combine{\ASA}{\Gen}(1^{\kappa})),\\ \am \in \{0,1\}^{\combine{\ASA}{\ml}(\kappa)},\\ k\in \supp(\combine{\SES}{\Gen}(1^{\kappa})),\\ m\in \{0,1\}^{\combine{\SES}{\ml}(\kappa)},\\ \sigma\in \{0,1\}^{*}}} \{\combine{\ASA}{\query}(\SES,\kappa,\ak,\am,k,m_j,\sigma)\}.
\end{align*}
For a family
$\mathcal{F}$ of encryption schemes, let
$\combine{\ASA}{\query}(\mathcal{F},\kappa)$ be the maximal value of
  $\combine{\ASA}{\query}(\SES,\kappa)$ for $\SES\in \mathcal{F}$. 

In the steganographic setting, \citeauthor{dedic2009upper} showed in
\cite{dedic2009upper} that (under the cryptographic assumption that
one-way functions exist) no
universal stegosystem can embed more than
$\mathcal{O}(1)\cdot \log(\kappa)$ bits per document and thus proved
that the rejection sampling based systems have optimal rate.
The needed ingredients of this proof are summarized 
by two key lemmas based on Lemma 12 and Lemma 13 in \cite{Berndt2016}.

\begin{lemma}
  \label{lem:secure_asa}
  Let $\ASA$ be a algorithm substitution attack for the symmetric encryption
  scheme $\SES$ such that
$\ASA$ is secure against $\SES$. 
Then for all integers $\kappa\in \nats$, 
messages $m\in \{0,1\}^{\combine{\ASA}{\ml}(\kappa)}$,  ciphertexts 
$c_{1},c_{2},\ldots,c_{\combine{\ASA}{\outl}(\kappa)}\gets
\combine{\ASA}{\Enc}(\ak,\am,k,m,\sigma)$ and all positions $i\in
\{1,\ldots,\combine{\ASA}{\outl}(\kappa)\}$: 
\begin{align*}
\Pr_{\ak\gets \combine{\ASA}{\Gen}(1^{\kappa})}[c_{i}\not\in \supp(\combine{\SES}{\Enc}(k,m))]
   \leq \InSec^{\encwatch}_{\ASA,\SES}(\kappa).
\end{align*}
\end{lemma}

\begin{lemma}
  \label{lem:reliable_asa}
    Let $\ASA$ be a universal and reliable algorithm substitution attack against
    the symmetric encryption scheme $\SES$. 
  Then for every $\kappa$, the probability that the encoder $\combine{\ASA}{\Enc}$ produces a ciphertext,
  which was not provided by the encryption oracle, is at least
   \begin{align*} 1-\unrel_{\ASA,\SES}(\kappa)-\frac{(\combine{\ASA}{\outl}(\kappa)\cdot \combine{\ASA}{\query}(\SES,\kappa))^{\combine{\ASA}{\outl}(\kappa)}}{2^{\combine{\ASA}{\ml}(\kappa)}}.
  \end{align*}
\end{lemma}

We will now show how one can modify an existing symmetric encryption scheme $\SES$ with
the help of a signature scheme $\SIG$ into a family of  encryption schemes
such that no universal ASA can achieve a super-logarithmic rate on
all of these encryption schemes. The construction is very similar to the
construction used in \cite{Berndt2016}.

A \emph{signature scheme}
$\SIG=(\combine{\SIG}{\Gen},\combine{\SIG}{\Sign},\combine{\SIG}{\Vrfy})$
is a triple of
probabilistic polynomial-time algorithms with the following properties:
\begin{itemize}
\item The \emph{key generator} $\combine{\SIG}{\Gen}$ produces upon input $1^{\kappa}$ a
  pair $(\pk,\sk)$ of keys with $|\pk|=|\sk|=\kappa$. We call $\pk$ the
  \emph{public key} and $\sk$ the \emph{secret key}.
\item The \emph{signing algorithm} $\combine{\SIG}{\Sign}$ takes as input the secret key
  $\sk$ and a message $m\in \{0,1\}^{\combine{\SIG}{\ml}}$ of length
  $\combine{\SIG}{\ml}(\kappa)$ and produces a signature
  $\sigma\in \{0,1\}^{\combine{\SIG}{\sigl}(\kappa)}$ of length
  $\combine{\SIG}{\sigl}(\kappa)$. 
\item The \emph{verifying algorithm} $\combine{\SIG}{\Vrfy}$ takes as input the public
  key $\pk$, the message $m$ and a signature $\sigma$ and outputs a bit
  $b$. 
\end{itemize}
We say that $(\Gen,\Sign,\Vrfy)$ is \emph{reliable}, if
$\Vrfy(\pk,m,\Sign(\sk,m))=1$ for all $\pk$, $\sk$ and $m$.

A \emph{forger} $\Forg$ is a probabilistic polynomial time algorithm
that upon input $\pk$ and oracle access to $\Sign_{\sk}$ tries to
produce a pair $(m,\sigma)$ such that
$\Vrfy_{\pk}(m,\sigma)=1$. Formally, this is defined via the following
experiment $\Sig-Forge$:

\myAlgorithm{$\Sig-Forge_{\Forg,\SIG}(\kappa)$}{length $\kappa$}{
  Forger $\Forg$, Signature
 Scheme $\SIG=(\Gen,\Sign,\Vrfy)$}{
\State $(\pk,\sk) \gets \Gen(1^{\kappa})$
\State $(m,\sigma) \gets \Forg^{\Sign_{\sk}}(\pk)$ 
\State Let $Q$ be the set of messages given to $\Sign_{\sk}$ by $\Forg$
\If{$m\not\in Q$ and $\Vrfy_{\pk}(m,\sigma)=1$}
   \AlgReturn{1}
\Else
 \AlgReturn{0}
\EndIf
}{Signature-Forging Experiment}

A signature scheme $\SIG$ is called \emph{existentially
  unforgeable}, if for every forger $\Forg$, there is a negligible
function $\negl$ such that
\begin{align*}
  \adv^{\sig}_{\Forg,\SIG}(\kappa) :=
  \Pr[\Sig-Forge_{\Forg,\SIG}(\kappa)=1]\leq \negl(\kappa).
\end{align*}

The maximal advantage of any forger against $\SIG$ is
called the \emph{insecurity} of $\SIG$ and is defined as
\begin{align*}
  \insec_{\SIG}^{\sig}(\kappa) := \max_{\Forg}\{\adv^{\sig}_{\Forg,\SIG}(\kappa)\}.
\end{align*}

For $(\pk,\sk)\in
\supp(\combine{\SIG}{\Gen}(1^{\kappa}))$, let $\SES_{\pk,\sk}$ be the
encryption scheme with 
\begin{itemize}
\item $\combine{\SES_{\pk,\sk}}{\Gen}=\combine{\SES}{\Gen}$, \ie the key
  generation algorithm remains the same.
\item The encryption algorithm $\combine{\SES_{\pk,\sk}}{\Enc}$ is given
  as:

  \inputAlgorithm{$\combine{\SES_{\pk,\sk}}{\Enc}$}{key $k$,
      message $m$}{
    \State $c \gets \combine{\SES}{\Enc}(k,m)$ 
    \State  $\gets \combine{\SIG}{\Sign}(\sk,c)$
    \State \AlgReturn{$(c,\sigma)$}}{Encryption Algorithm}
\item Similarly, the decryption algorithm
  $\combine{\SES_{\pk,\sk}}{\Dec}$ is given as:

  \inputAlgorithm{$\combine{\SES_{\pk,\sk}}{\Dec}$}{key $k$,
      ciphertext $(c,\sigma)$}{
    \If {$\combine{\SIG}{\Vrfy}(\pk,c,\sigma)=1$}
    \State \AlgReturn{$\combine{\SES}{\Dec}(k,c)$}
    \Else \AlgReturn{$\bot$}
    \EndIf}{Decryption Algorithm}

\end{itemize}

By using this family
\begin{align*}
\mathcal{F}(\SES,\SIG)=\{\SES_{\pk,\sk}\}_{(\pk,\sk)\in     \supp(\combine{\SIG}{\Gen}(1^{\kappa}))},
\end{align*}

we can derive the following
upper bound on the rate of each universal ASA:
\begin{theorem}
  Let $\SES$ be a symmetric encryption scheme, $\SIG$
  be a signature scheme and $\mathcal{F}=\mathcal{F}(\SES,\SIG)$ be defined as
  above. For every universal algorithm substitution attack $\ASA$  against $\SES$, 
 there exist a forger $\Forg$ on $\SIG$ with
  advantage at least
  \begin{align*}
1-\InSec^{\encwatch}_{\ASA,\mathcal{F}}(\kappa)-\unrel_{\ASA,\mathcal{F}}(\kappa)-\varphi(\ASA,\kappa)
  \end{align*}
  for every $\kappa$, where
  \begin{align*}
    \varphi(\ASA,\kappa)=\frac{(\combine{\ASA}{\outl}(\kappa)\cdot \combine{\ASA}{\query}(\mathcal{F},\kappa))^{\combine{\ASA}{\outl}(\kappa)}}{2^{\combine{\ASA}{\ml}(\kappa)}}.
  \end{align*}
\end{theorem}

\begin{proof}
  The proof is analogue to the proof of \cite[Theorem 13]{Berndt2016}.
  
    Fix $\kappa\in \nats$ and $(\pk,\sk)\in \supp(\combine{\SIG}{\Gen}(1^{\kappa}))$. We
  will now construct an forger on $\SIG$ with the help of the algorithm
  substitution attacker
  $\ASA$. Choose a random attacker message $\am^{*}\gets
  \{0,1\}^{\combine{\ASA}{\ml}(\kappa)}$, 
  a random attacker key $\ak^{*}\gets \combine{\ASA}{\Gen}(1^{\kappa})$,
  a random message $m^{*}\gets  \{0,1\}^{\combine{\SES}{\ml}(\kappa)}$
  and a random key $k^{*}\gets \combine{\SES}{\Gen}(1^{\kappa})$.
  
  The forger now simulates the run of the algorithm substitution attack 
  $\combine{\ASA}{\Enc}^{\combine{\SES_{\pk,\sk}}{\Enc}(\cdot,\cdot;\cdot)}(\ak^{*},\am^{*},k^{*},m^{*})$
  against the symmetric encryption scheme 
  $\SES_{\pk,\sk}$. Whenever $\combine{\ASA}{\Enc}$ makes an access
  $(k,m;r)$ to its encryption oracle, the forger computes
  $c=\combine{\SES}{\Enc}(k,m;r)$ and uses its
  signing oracle $\combine{\SIG}{\Sign}_{\sk}$ upon $c$. This
  returns a valid signature $\sigma$ for $c$ and the forger returns
  $(c,\sigma)$ to $\combine{\ASA}{\Enc}$. 
  This simulation hence yields the same result as
  $\combine{\ASA}{\Enc}^{\combine{\SES_{\pk,\sk}}{\Enc}(\cdot,\cdot;\cdot)}(\ak^{*},\am^{*},k^{*},m^{*})$.
  Denote the first document
  produced by the run of the algorithm substitution attack 
  $\combine{\ASA}{\Enc}^{\combine{\SES_{\pk,\sk}}{\Enc}(\cdot,\cdot;\cdot)}(\ak^{*},\am^{*},k^{*},m^{*})$
  as 
  $(\widehat{c},\widehat{\sigma})$. By \autoref{lem:secure_asa}, the
  probability that the pair 
  $(\widehat{c},\widehat{\sigma})$ does not belong to to the support $\supp(\combine{\SES_{\pk,\sk}}{\Enc}(k,m))$
  (\ie it is no valid ciphertext-signature pair) is bounded by
  $\InSec^{\encwatch}_{\ASA,\SES_{\pk,\sk}}(\kappa)$.
  Furthermore, \autoref{lem:reliable_asa}
  implies that the probability that $(\widehat{c},\widehat{\sigma})$ is equal to any
  $(c,\sigma)$ which was given to the ASA is at most
$\unrel_{\ASA,\SES_{\pk,\sk}}(\kappa)+\varphi(\ASA,\kappa)$. 
We can thus conclude that 
  with probability
  \begin{align*}
    1-&\InSec^{\encwatch}_{\ASA,\SES_{\pk,\sk}}(\kappa)-\unrel_{\ASA,\SES_{\pk,\sk}}(\kappa)-\\
      &\frac{(\combine{\ASA}{\outl}(\kappa)\cdot \combine{\ASA}{\query}(\SES_{\pk,\sk},\kappa))^{\combine{\ASA}{\outl}(\kappa)}}{2^{\combine{\ASA}{\ml}(\kappa)}},
  \end{align*}
  the ciphertext-signature pair $(\widehat{c},\widehat{\sigma})$ is a valid
  ciphertext-signature pair and was not produced by the oracle $\combine{\SIG}{\Sign}_{\sk}$
 The advantage of the forger against the
  signature scheme $\SIG$ is thus at least
  \begin{align*}
    1-&\InSec^{\encwatch}_{\ASA,\SES_{\pk,\sk}}(\kappa)-\unrel_{\ASA,\SES_{\pk,\sk}}(\kappa)-\\
    &\frac{(\combine{\ASA}{\outl}(\kappa)\cdot \combine{\ASA}{\query}(\SES_{\pk,\sk},\kappa))^{\combine{\ASA}{\outl}(\kappa)}}{2^{\combine{\ASA}{\ml}(\kappa)}},
  \end{align*}
  The
  running time of the forger is polynomial in $\kappa$ due to the
  polynomial running time of $\combine{\ASA}{\Enc}$.
\end{proof}
This allows us to conclude the following corollary bounding the number
of bits embeddable into a single ciphertext by a universal algorithm
substitution attack. 
\begin{corollary}
\label{cor:bound}
  There is no universal algorithm substitution attack that embeds
  more than $\mathcal{O}(1)\cdot \log(\kappa)$ bits per ciphertext
  (unless one-way functions do not exist).
\end{corollary}

\section{Conclusions}
In this work, we proved that ASAs in the strong undetectability model of
Bellare, Jaeger and Kane~\cite{bellare2015asa} are a special case of
stegosystems on a certain kind of channels described by symmetric
encryption schemes. This gives a rigorous proof of the well-known
connection between steganography and algorithm substitution attacks. We
make use of this relationship to show that a wide range of results on ASAs
are already present in the steganographic literature. Inspired by this
connection, we define \emph{universal ASAs} that work with no knowledge
on the internal implementation of the symmetric encryption schemes and
thus work for \emph{all} such encryption schemes with sufficiently large
min-entropy. As almost all known ASAs are universal, we investigate
their rate -- the number of embedded bits per ciphertext -- and prove
a logarithmic upper bound of this rate.

\bibliographystyle{ACM-Reference-Format}


\end{document}